\documentclass[11pt,letterpaper]{article}

\usepackage[utf8]{inputenc}
\usepackage[margin=1in]{geometry}
\usepackage{centernot}
\usepackage[dvipsnames,usenames]{xcolor}
\usepackage[colorlinks=true,pdfpagemode=UseNone,urlcolor=RoyalBlue,linkcolor=RoyalBlue,citecolor=OliveGreen,pdfstartview=FitH]{hyperref}
\usepackage{amsthm}
\usepackage{amsmath}
\usepackage{amssymb}
\usepackage{bm}

\usepackage{booktabs}
\usepackage{subcaption}

\usepackage{tikz}

\usepackage[sort,numbers]{natbib}

\usepackage{todonotes}
\usepackage{verbatim}
\usepackage{nicefrac}

\title{
    Strong Revenue (Non-)Monotonicity of Single-parameter Auctions
    \thanks{
        This is the second version of the paper on arXiv.
        Compared to the first version, this version extends approximate strong revenue monotonicity to a uniform notion of closeness of value distributions, and includes two new lower bounds (Theorems~\ref{thm:revenue-lipschitz-lower-bound} and \ref{thm:sample-complexity-general-lower-bound}) that match the corresponding upper bounds up to logarithmic factors.
        We thank anonymous reviewers for their insightful questions which lead to these new results.
    }
}
\author{
    Ziyun Chen
    \thanks{IIIS, Tsinghua University. Email: chenziyu20@mails.tsinghua.edu.cn.}
    \and
    Zhiyi Huang
    \thanks{The University of Hong Kong. Email: zhiyi@cs.hku.hk, u3563782@connect.hku.hk.}
    \and
    Dorsa Majdi
    \thanks{Sharif University of Technology. Email: dorsa.majdi@gmail.com}
    \and
    Zipeng Yan
    \footnotemark[2]
}
\date{}

\newtheorem{theorem}{Theorem}[section]
\newtheorem{corollary}[theorem]{Corollary}
\newtheorem{lemma}[theorem]{Lemma}

\theoremstyle{definition}
\newtheorem{definition}{Definition}[section]
\newtheorem{remark}{Remark}[section]

\newcommand{\uniformclose}{\stackrel{\mbox{\tiny unif.}}{\approx}}

\newcommand{\defeq}{\stackrel{\mbox{\tiny def.}}{=}}

\newcommand{\allocset}{\ensuremath{\mathcal{X}}}
\newcommand{\alloc}{\ensuremath{x}}

\newcommand{\auction}{M}
\newcommand{\myerson}{M}
\newcommand{\virtual}{\varphi}
\newcommand{\OPT}{\textsc{Opt}}
\newcommand{\ALG}{\textsc{Alg}}
\newcommand{\Dif}{\textsc{Dif}}

\newcommand{\indset}{\mathcal{I}}

\newcommand{\R}{\ensuremath{\mathbb{R}}}

\newcommand{\TotalVar}{\ensuremath{\text{\rm TV}}}
\newcommand{\Hellinger}{\ensuremath{\text{\rm H}}}

\DeclareMathOperator{\E}{\mathbf{E}}
\let\Pr\relax
\DeclareMathOperator{\Pr}{\mathbf{Pr}}

\newcommand{\bq}{q}
\newcommand{\bx}{x}
\newcommand{\btildex}{\tilde{x}}
\newcommand{\bv}{v}
\newcommand{\btildev}{\tilde{v}}

\newcommand{\bD}{D}
\newcommand{\btildeD}{\tilde{D}}
\newcommand{\bhatD}{\hat{D}}
\newcommand{\bE}{E}
\newcommand{\btildeE}{\tilde{E}}
\newcommand{\bP}{P}
\newcommand{\bQ}{Q}

\newcommand{\zh}[1]{\todo[color=green!40,inline]{Zhiyi: #1}}

\newcommand{\yan}[1]{\todo[color=blue!40,inline]{Yan: #1}}

\begin{document}

\begin{titlepage}
    \thispagestyle{empty}
    \maketitle
    \begin{abstract}
        \thispagestyle{empty}
        Consider Myerson's optimal auction with respect to an inaccurate prior, e.g., estimated from data, which is an underestimation of the true value distribution.
Can the auctioneer expect getting at least the optimal revenue w.r.t.\ the inaccurate prior since the true value distribution is larger?
This so-called \emph{strong revenue monotonicity} is known to be true for single-parameter auctions when the feasible allocations form a matroid.
We find that strong revenue monotonicity fails to generalize beyond the matroid setting, and further show that auctions in the matroid setting are the only downward-closed auctions that satisfy strong revenue monotonicity.
On the flip side, we recover an approximate version of strong revenue monotonicity that holds for all single-parameter auctions, even without downward-closedness.
As applications, we get sample complexity upper bounds for single-parameter auctions under matroid constraints, downward-closed constraints, and general constraints.
They improve the state-of-the-art upper bounds and are tight up to logarithmic factors.

    \end{abstract}

\end{titlepage}

\section{Introduction}
\label{sec:intro}

Revenue optimal auction design is a central topic in economics and more recently in algorithmic game theory.
For example, consider auctioning an item to some bidders.
How shall the auctioneer decide which bidder wins the item and how much the winner pays based on the bids, so that bidders would truthfully report their values for the item and the auctioneer's revenue is maximized?
Classical auction theory often studied this problem under the Bayesian model, in which the bidders values are drawn from some value distribution known to the auctioneer beforehand.
The goal is to maximize the expected revenue over the random realization of bidders' values.

The simplest case is when there is only one bidder.
The problem then becomes choosing a take-it-or-leave-it price $p$ to maximize the product of price $p$ and the probability that the bidder's value is at least $p$.

The revenue optimal auction is more involved when there are multiple bidders.
\citet{Myerson:MOR:1981} characterized the revenue optimal auction when the bidders' values are independently (but not necessarily identically) distributed;
the literature often refers to it as Myerson's optimal auction.
In a nutshell, the auctioneer computes for each bidder an \emph{ironed virtual value} based on the bidder's value/bid and value distribution.
Then, the auctioneer allocates the item to the bidder with the largest nonnegative virtual value, and leaves the item unallocated if all bidders have negative virtual values.
Finally, the winner of the item pays the \emph{threshold value}, i.e., the smallest value at which it would still win.

Myerson's optimal auction generalizes to all \emph{single-parameter auctions}, which will be the focus of this paper.
For ease of exposition, we consider the following simplified definition of single-parameter auctions in the introduction;
Section~\ref{sec:prelim} will give a more general definition.
Suppose that the auctioneer has some homogeneous items which it can allocate to certain subsets of bidders.
We refer to each subset of bidders that the auctioneer can allocate to as a \emph{feasible allocation}.
Below are some example constraints that may define the set of feasible allocations:
\begin{enumerate}
    \item[(C1)] There are $k$ copies of the item and therefore at most $k$ bidders can be allocated an item.
    \item[(C2)] Two bidders are competitors so the auctioneer can not allocate to both at the same time.
    \item[(C3)] Two bidders are bundled in the sense that either none or both are allocated an item.
\end{enumerate}


Each bidder has a value for being allocated an item, drawn from a value distribution.
In this sense, its value is given by a single parameter and hence the name of the setting.
The auctioneer knows the value distribution but not the realized value.
Myerson's optimal auction first asks all bidders to submit bids.
Then, treating each bidder's bid as its value, the auctioneer computes the ironed virtual value of each bidder, chooses an allocation that maximizes the sum of ironed virtual values of the allocated bidders, and finally lets each allocated bidder pays its threshold value.

\subsection{Inaccurate Prior and Revenue Monotonicity}

To faithfully implement Myerson's optimal auction, the auctioneer would need to have complete information of each bidder's value distribution, which is rarely available.
Further, it is known that the ironed virtual values are sensitive to smaller changes to the value distributions (e.g., \cite{GuoHZ:STOC:2019, RoughgardenS:EC:2016}).
This leads to a natural question:
\emph{Can the auctioneer use Myerson's optimal auction with respect to (w.r.t.) an inaccurate prior, e.g., estimated from various data, and still expect good revenue?}

Let us first revisit the simplest case of a single bidder.
It is folklore that underestimating and overestimating the value distribution and, correspondingly, decreasing and increasing the resulting take-it-or-leave-it prices have contrasting impacts to the revenue.
Decreasing the price by $1\%$ would at worst lower the expected revenue by $1\%$, while increasing the price by $1\%$ might lower the expected revenue to almost zero.
The robustness to underestimation could be formalized as an \emph{auction-wise revenue monotone} property.
Consider any value distributions $D$ and $\tilde{D}$ such that the former (stochastically) dominates the latter, i.e., the former's cumulative distribution function (CDF) is point-wise less than or equal to the latter's CDF.
The revenue of any auction (i.e., any take-it-or-leave-it price) w.r.t.\ $D$ is greater than or equal to the revenue w.r.t.\ $\tilde{D}$.
Hence, it is safer to use the optimal take-it-or-leave-it price w.r.t.\ an inaccurate prior $\tilde{D}$ that is an underestimation, since the auctioneer can \emph{guarantee getting at least the estimated revenue of the chosen auction}.

Do similar properties hold when there are multiple bidders?
It is easy to construct counter-examples that refute auction-wise revenue monotonicity;%
\footnote{For example, consider two bidders Alice and Bob and the auctioneer can allocate to at most one. Consider an auction that allocates to Alice and charges her $\$1$ if her value is at least $\$1$, and otherwise allocates to Bob and charges him $\$10$ if his value is at least $\$10$. When Alice's value and Bob's value are deterministically $\$0$ and $\$10$, the revenue is $\$10$. Increasing Alice's value distribution to be deterministically $\$1$, however, lowers the revenue to $\$1$.}
but is it still safe to use the optimal auction w.r.t.\ an underestimation and guarantee getting at least the estimated revenue?
\citet*{DevanurHP:STOC:2016} gave an affirmative answer when the set of feasible allocations form a matroid.
They called this property \emph{strong revenue monotonicity}.
While we defer the definition of matroids to Section~\ref{sec:prelim}, readers may think of the aforementioned constraint (C1) as a running example, which is called a $k$-uniform matroid.

Beyond the matroid setting, only a weaker notion of \emph{revenue monotonicity} was known (folklore, c.f., \cite{DevanurHP:STOC:2016}):
if value distribution $\bD$ dominates $\btildeD$ then the optimal revenue w.r.t.\ $\bD$ is greater than or equal to the optimal revenue w.r.t.\ $\btildeD$.
This weaker notion is insufficient for answering our motivating question, i.e., how much revenue the auctioneer could expect when it uses Myerson's optimal auction w.r.t.\ an inaccurate prior.
Hence, this paper studies whether strong revenue monotonicity holds for general single-parameter auctions.
Before getting to our results, we remark that intriguingly even the weaker notion of revenue monotonicity ceases to hold in the presence of multiple types of items, a.k.a., the multi-parameter setting, as proved by \citet{HartR:TE:2015}.

\paragraph{Our Contribution.}
On the one hand, we prove that strong revenue monotonicity does not hold in general single-parameter auctions (Theorem~\ref{thm:non-monotone}), even if the set of feasible allocations is \emph{downward-closed}, i.e., removing a bidder from any feasible allocation would give another feasible allocation, such as the aforementioned constraints (C1) and (C2).
In fact, we show that a downward-closed single-parameter auction is strongly revenue monotone if and only if its feasible allocations form a matroid (Theorem~\ref{thm:strong-revenue-monotone-iff-matroid}).
Further, the decrease in revenue could be as large as a constant fraction of the optimal revenue (Corollary~\ref{cor:non-monotone}).

On the other hand, we show that the auctioneer can nonetheless ensure good revenue by using Myerson's optimal auction w.r.t.\ an \emph{approximately accurate} underestimation for all single-parameter auctions, including those that are not downward-closed, e.g., with the aforementioned constraint (C3).
Concretely, we prove that if $\bD$ dominates $\btildeD$ and further $\bD$ and $\btildeD$ are \emph{sufficiently close}, then running Myerson's optimal auction w.r.t.\ $\btildeD$ when the value distribution is $\bD$ gets almost the optimal revenue w.r.t.\ $\btildeD$ (Theorem~\ref{thm:approximate-revenue-monotone}).
We call this \emph{approximate strong revenue monotonicity}.

Figure~\ref{fig:monotone} summarizes the revenue monotone properties in different auction settings.

\begin{figure}[t]
    \centering
    \begin{tikzpicture}[scale=0.9, every node/.style={scale=0.9}]
        \draw[fill=red!40] {(4.5,0) ellipse (6.4 and 2.8)};
        \node[align=center] at (9.7,-.05) {Multi-\\parameter};
        \draw[fill=red!40,draw=none] (9.2,-4) rectangle +(4,1.8) node[midway,align=center] {Not monotone\\(\citet{HartR:TE:2015})};
        \draw[ultra thick,-,dashed,Red] (11.2,-2.2)--(11.2,-0)--(10.9,0);
        %
        \draw[fill=orange!40] {(3.4,0) ellipse (5.2 and 2.3)};
        \node[align=center] at (7.3,-.05) {Single-\\parameter};
        %
        \draw {(2.3,0) ellipse (4 and 1.8)};
        \node[align=center] at (4.8,0) {Downward-\\closed};
        \draw[fill=orange!40,draw=none] (3.5,3.6) rectangle +(9.7,1.8) node[midway,align=left] {Monotone (folklore)\\ {\bf Not strongly monotone (Sec.~\ref{sec:non-monotone})}\\ {\bf Approximately strongly monotone (Sec.~\ref{sec:approx-monotone})}};
        \draw[ultra thick,-,dashed,Orange] (8.4,3.6)--(8.4,3.1)--(3.4,3.1)--(3.4,2.3);
        %
        \draw[fill=yellow!35] {(1,0) ellipse (2.6 and 1.4)};
        \node at (2.5,0.04) {Matroid};
        \draw[fill=yellow!35,draw=none] (-4,3.6) rectangle +(6.6,1.8) node[midway,align=center] {Strongly monotone\\(\citet*{DevanurHP:STOC:2016})};
        \draw[ultra thick,-,dashed,Dandelion] (-0.7,3.6)--(-0.7,3.1)--(1,3.1)--(1,1.4);
        %
        \draw[fill=green!30] {(0,0) ellipse (1.5 and 0.9)};
        \node at (0,0) {Single-bidder};
        \draw[fill=green!30,draw=none] (-4,-4) rectangle +(3.2,1.8) node[midway,align=center] {Auction-wise\\ monotone\\ (folklore)};
        \draw[ultra thick,-,dashed,Green] (-2.4,-2.2)--(-2.4,0)--(-1.5,0);
    \end{tikzpicture}
    \caption{Revenue monotone properties in different auction settings. A setting satisfies a revenue monotone property if all auctions therein satisfies the property. A setting does not satisfy a revenue monotone property if there is an auction in the setting violating the property.}
    \label{fig:monotone}
\end{figure}
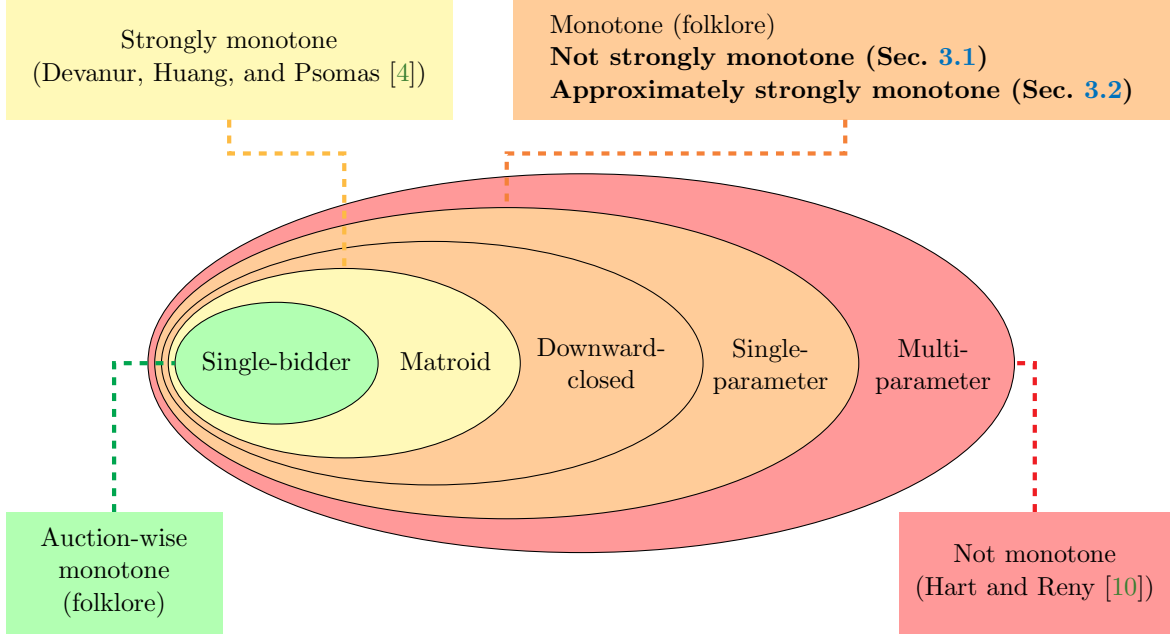

\subsection{Sample Complexity}

Closely related to the analysis of Myerson's optimal auction obtained from an inaccurate prior, \citet{ColeR:STOC:2014} introduced a model in which the auctioneer can only access the value distribution through i.i.d.\ samples.
They asked how many samples are sufficient and necessary for learning an auction that is optimal up to an $\epsilon$ error?%
\footnote{\citet{ColeR:STOC:2014} studied regular value distributions and $(1-\epsilon)$-multiplicative approximation to the revenue of Myerson's optimal auction w.r.t.\ the true value distribution. We consider value distributions with bounded supports $[0, 1]$ and $\epsilon$-additive approximation. Nonetheless, regular distributions and multiplicative approximation and several other settings can be reduced to our setting through appropriate discretizations (see, e.g., \cite{GuoHZ:STOC:2019, GonczarowskiN:STOC:2017}).}
Driven by strong revenue monotonicity, \citet*{GuoHZ:STOC:2019} proposed to use Myerson's optimal auction w.r.t.\ a \emph{dominated product empirical} distribution derived from samples, which is dominated by the true distribution and is as close to the true distribution as possible.
They showed that this approach gives sample complexity upper bounds that are tight up to logarithmic factors for auctions under matroid constraints.

\paragraph{Our Contribution.}
We prove that the approximate strong revenue monotonicity proposed in this paper is good enough for deriving sample complexity upper bounds using Myerson's optimal auction w.r.t.\ the dominated product empirical.
In fact, our analysis is an improvement over that of \citet{GuoHZ:STOC:2019} and therefore even in the matroid setting our $O(\frac{nk}{\epsilon} \log \frac{nk}{\epsilon\delta})$ upper bound (Theorem~\ref{thm:sample-complexity}) is better than theirs by three logarithmic factors.
Here $n$ is the number of bidders, $k$ is the maximum number of bidders that can be allocated to in any feasible allocation, a.k.a., the \emph{rank}, and $\delta$ is the probability that the algorithm fails to obtain an $\epsilon$-additive approximation.
For downward-closed auctions, we derive the same bound as in the matroid setting.
It improves the best previous bound by \citet{GonczarowskiN:STOC:2017} by a multiplicative $\frac{k}{\epsilon}$ factor, and is tight up to a logarithmic factor due to the known lower bound in the more special matroid setting~\cite{GuoHZ:STOC:2019}.
Finally, for arbitrary single-parameter auctions, we obtain an upper bound of $O(\frac{nk^2}{\epsilon^2} \log \frac{nk}{\epsilon} \log \frac{nk}{\epsilon\delta})$ (Theorem~\ref{thm:sample-complexity-general}), which improves the best previous bound by \citet{GonczarowskiN:STOC:2017} by a multiplicative $\frac{1}{\epsilon}$ factor.
We also prove a lower bound of $\Omega(\frac{nk^2}{\epsilon^2})$ (Theorem~\ref{thm:sample-complexity-general-lower-bound}), matching the upper bound up to logarithmic factors.
Our results further demonstrate that general single-parameter auctions are intrinsically harder than downward-closed auctions in terms of sample complexity.


\begin{table}[t]
    \caption{Summary of sample complexity upper bounds of different single-parameter auctions. Here $n$ denotes the number of bidders, $k$ denotes the maximum total allocated amount in any feasible allocation, a.k.a., the rank, $\epsilon$ denotes the additive approximation factor, and $\delta$ denotes the algorithm's failure probability. The bounds from this paper are tight up to logarithmic factors.}
    \label{tab:sample-complexity}
    \centering
    \renewcommand{\arraystretch}{1.2}
    \renewcommand{\thefootnote}{\alph{footnote}}
    \begin{tabular}{l@{\hskip 20pt}ll@{\hskip 20pt}ll}
    \toprule
    & \multicolumn{2}{l}{\hskip -5pt Best Previous Bound} & \multicolumn{2}{l}{\hskip -5pt \textbf{This Paper}} \\
    \midrule
    Single-bidder & $O(\frac{1}{\epsilon^2} \log \frac{1}{\epsilon\delta})$ & \cite{HuangMR:SICOMP:2018} & \\
    Matroid & $O(\frac{nk}{\epsilon^2} \log^2\frac{n}{\epsilon} \log\frac{nk}{\epsilon} \log\frac{nk}{\epsilon\delta})$ & \cite{GuoHZ:STOC:2019} & $O(\frac{nk}{\epsilon^2} \log\frac{nk}{\epsilon\delta})$ & \textbf{(Thm.~\ref{thm:sample-complexity})} \\
    Downward-closed & $O(\frac{nk^2}{\epsilon^3} \log\frac{nk}{\epsilon\delta})$ & \cite{GonczarowskiN:STOC:2017} & $O(\frac{nk}{\epsilon^2} \log\frac{nk}{\epsilon\delta})$ & \textbf{(Thm.~\ref{thm:sample-complexity})}\\
    Single-parameter & $O(\frac{nk^2}{\epsilon^3} \log\frac{nk}{\epsilon\delta})$ & \cite{GonczarowskiN:STOC:2017} & $O(\frac{nk^2}{\epsilon^2} \log\frac{nk}{\epsilon} \log\frac{nk}{\epsilon\delta})$ & \textbf{(Thm.~\ref{thm:sample-complexity-general})}\\
    \bottomrule
    \end{tabular}
\end{table}

\subsection{Related Works}

For multi-parameter auctions, \citet{HartR:TE:2015} gave an example of revenue non-monotonicity involving only one bidder and two heterogeneous items.
They also showed that two special classes of single-bidder multi-item auctions have monotone payment functions and as a result satisfy strong revenue monotonicity.%
\footnote{\citet{HartR:TE:2015} only claimed revenue monotonicity but implicitly proved strong revenue monotonicity as well.}
There is a long line of works proving that simple auctions can guarantee nearly optimal revenue in various multi-parameter auctions, e.g., by bundling all items together and posting a take-it-or-leave-it price, and by selling items separately (e.g., \cite{HartN:JET:2017, LiY:PNAS:2013, BabaioffILW:FOCS:2014, RubinsteinW:TEAC:2018, CaiZ:STOC:2017}).
\citet{RubinsteinW:TEAC:2018} observed that these simple single-bidder auctions satisfy revenue monotonicity, i.e., they yield better revenue on stochastically dominating distributions, and therefore the optimal revenue of the respective multi-parameter auctions is approximately monotone.
\citet{Yao:SAGT:2018} extended the approximate revenue monotonicity to multiple bidders with fractionally subadditive valuations.

Following \citet{ColeR:STOC:2014}, there has been a vast literature devoted to the sample complexity of various auctions.
\citet*{HuangMR:SICOMP:2018} resolved the sample complexity of the single-bidder case up to logarithmic factors.
\citet{MorgensternR:NIPS:2015}, \citet{DevanurHP:STOC:2016}, \citet{GonczarowskiN:STOC:2017}, and \citet{Syrgkanis:NIPS:2017} built on learning theory to improve the sample complexity upper bound of single-parameter auctions with multiple bidders.
\citet{GuoHZ:STOC:2019} built on strong revenue monotonicity and got sample complexity upper and lower bounds tight up to logarithmic factors for the matroid setting.
Although a complete characterization of optimal multi-parameter auctions remains elusive, \citet{GonczarowskiW:JACM:2021} and \citet{GuoHTZ:COLT:2021} showed that polynomially many samples are informationally sufficient for learning a multi-parameter auction optimal up to an $\epsilon$ error.
Last but not least, \citet{GuoHTZ:COLT:2021} extended the notion of strong monotonicity to other Bayesian optimization problems including prophet inequality and Pandora's problem, and obtained nearly tight sample complexity upper bounds for them.

\section{Preliminaries}
\label{sec:prelim}

\paragraph{Notations.}
Let $\R_+$ denote the set of nonnegative real numbers.
Let $[n] = \{1, 2, \dots, n\}$.
For any distribution $D$ on $\R_+$, we abuse notation and let $D$ also denote its CDF, i.e., $D(v) = \Pr_{u \sim D} [ u \le v ]$;
thus, its derivative $D'(v)$ (if exists) is the probability density function (PDF).
For distributions $D$ and $\tilde{D}$ on $\R_+$, we say that $D$ (first-order) stochastically dominates $\tilde{D}$, denoted as $D \succeq \tilde{D}$, if $D(v) \le \tilde{D}(v)$ for all $v \in \R_+$.
For any distribution $D$ over a domain $\Omega$ and any function $f : \Omega \to \R$, we write $f(D)$ for the expected function value $\E_{\omega \sim D} f(\omega)$.

\subsection{Single-parameter Auctions}

In a single-parameter auction with $n$ bidders, each bidder $i$ has a \emph{private} value $0 \le v_i \le 1$ drawn independently from a distribution $D_i$.
An auction proceeds as follows.
First each bidder $i$ submits a bid $b_i$ to the auctioneer.
The auctioneer then picks an allocation $\alloc$ from a set $\allocset \subset \R_+^n$ of feasible allocations, and prices $p \in \R_+^n$ according to the bids.
Note that from now on we consider a more general model that allows the allocation $\alloc$ to be non-binary.
Each bidder $i$ receives allocation $\alloc_i$, pays $p_i$, and gets utility $v_i \alloc_i - p_i$.
Throughout the paper we will focus on \emph{truthful} auctions in which a bidder can always guarantee a non-negative utility and can maximize its utility by bidding its value.
Hence, we will assume $b_i = v_i$ and will no longer talk about bids.

Let $k = \max_{\alloc \in \allocset} \| \alloc \|_1$ be the maximum size of any feasible allocation.
Following a terminology from the special case when $\allocset$ is the convex hull of a matroid, we refer to $k$ as the \emph{rank}.
Finally, for ease of exposition we assume without loss of generality (WLOG) that the problem is \emph{unit-demand}, i.e., $\alloc_i \le 1$ for any feasible allocation $\alloc \in \allocset$ and any bidder $i$.
The general case can be reduced to this unit-demand case, since an $\epsilon$-approximation in an auction with $n$ bidders, rank $k$, and maximum demand $d$ is equivalent to an $\frac{\epsilon}{d}$-approximation in another auction with $n$ bidders, rank $\frac{k}{d}$, and unit-demand, obtained by scaling all allocations by a factor $\frac{1}{d}$.


\paragraph{Quantiles and Revenue Curves.}
For any value distribution $D$ on $\R_+$, the quantile of a value $v \in \R_+$ w.r.t.\ $D$ is the probability that a sample from $D$ is greater than $v$:
\[
    q_D(v) = \Pr_{u \sim D} \big[ u > v \big]
    ~.
\]
On the other hand, for any quantile $q \in [0, 1]$ the corresponding value w.r.t.\ $D$ is:
\[
    v_D(q) = \inf \big\{ v : q_D(v) \le q \big\} 
    ~.
\]
We remark that $v_D$ is simply the inverse of $q_D$ if $D$ is continuous.

For any value distribution $D$ and any quantile $0 \le q \le 1$, consider $q^+ = \Pr_{u \sim D} [u \ge v_D(q)]$ and $q^- = \Pr_{u \sim D}[u > v_D(q)]$ in the next definition.
If the distribution is continuous, we would have that $q = q^+ = q^-$.
The revenue curve w.r.t.\ value distribution $D$ in the \emph{quantile space} is:
\[
    R_D(q) =
    \begin{cases}
        q \cdot v_D(q)
        & \mbox{if $D$ is continuous at $v_D(q)$;} \\
        \frac{q-q^-}{q^+-q^-} q^+ v_D(q^+) + \frac{q^+-q}{q^+-q^-} q^- v_D(q^-)
        & \mbox{if $v_D(q)$ is a point mass.}
    \end{cases}
\]

The \emph{ironed} revenue curve is its convex hull, i.e.:
\[
    \bar{R}_D(q) = \max \big\{ ~ R_D(P) : \mbox{$P$ is a distribution over $[0,1]$ with expectation $q$} ~ \big\}
    ~.
\]

\paragraph{Optimal Auction.}
Given any truthful auction $\auction$, abuse notation and let $\auction$ also be a mapping from value profiles to the resulting revenue.
That is, for any $\bv = (v_1, v_2, \dots, v_n)$, $\auction(\bv)$ denotes the revenue of $\auction$ when bidders bid $\bv$;
for any product value distribution $\bD = D_1 \times D_2 \times \dots \times D_n$, $\auction(\bD)$ denotes the expected revenue of $\auction$.
Let $\OPT(\bD)$ be the largest expected revenue achievable by truthful auctions when the value distribution is $\bD$.

\citet{Myerson:MOR:1981} introduced the \emph{virtual values} of any bidder $i$, defined as $\virtual_i(v_i) = v_i - \frac{1-D_i(v)}{D_i'(v)}$ when value distribution $D_i$ is continuous.
For general value distributions, the virtual value is the right derivative of revenue curve $R_{D_i}(q)$ at $q_D(v_i)$.
Myerson showed that an auction is truthful if and only if its allocation rule $x$ is monotone, i.e., if $x_i$ is nondecreasing in $v_i$ for all bidders $i$, and the payment rule $p$ is determined by a specific formula according to the allocation rule $x$.
Further, the expected revenue equals the expected \emph{virtual welfare}:
\[
    \E_{\bv \sim \bD} \sum_{i=1}^n x_i(\bv) \virtual_i(v_i)
    ~.
\]

Myerson further defined the ironed virtual values of any bidder $i$, denoted as $\bar{\virtual}_i$.
It is the right derivative of the ironed revenue curve $\bar{R}_{D_i}(q)$ at $q_D(v_i)$.
Finally, Myerson showed that to maximize the expected revenue, the optimal auction always chooses an allocation that maximizes the ironed virtual welfare.
Denote Myerson's optimal auction w.r.t.\ a value distribution $\bD$ as $\myerson_{\bD}$.

A value distribution is \emph{regular} if its revenue curve is concave.
For regular distributions, the ironed revenue curve coincides with the revenue curve, and the ironed virtual values coincide with the virtual values.

Figure~\ref{fig:revenue-curve} presents an illustrative example of the revenue curves and virtual values w.r.t.\ a value distribution that has a point mass of $\nicefrac{1}{5}$ at value $\nicefrac{1}{2}$, and is uniform over $[0, 1]$ otherwise.
Figure~\ref{fig:revenue-curve-continous} demonstrates the case of value $v = \nicefrac{1}{5}$ which is not a point mass; it corresponds to quantile $q = \nicefrac{4}{5}$, and the right derivative of the revenue curve equals the virtual value.
Figure~\ref{fig:revenue-curve-point-mass} shows the case of value $v = \nicefrac{1}{2}$ which is a point mass; it corresponds to a \emph{left-closed-right-open} quantile interval $q \in [\nicefrac{2}{5}, \nicefrac{3}{5})$, and the right derivative of the revenue curve equals the virtual value.
Finally, Figure~\ref{fig:ironed-revenue-curve} gives the ironed revenue curve, ironing quantile interval $[\nicefrac{(3-\sqrt{3})}{5},\nicefrac{3}{5})$.
It effectively rounds the values from $\nicefrac{1}{2}$ to $\nicefrac{(1+\sqrt{3})}{4}$ down to $\nicefrac{1}{2}$;
the ironed revenue curve is the revenue curve of the rounded distribution, with a point mass of $\nicefrac{\sqrt{3}}{5}$ at value $\nicefrac{1}{2}$, and has probability density $\nicefrac{4}{5}$ in value intervals $[0, \nicefrac{1}{2})$ and $(\nicefrac{(1+\sqrt{3})}{4}, 1]$.
We remark that interpreting ironing as a rounding of values will be a useful viewpoint in an argument in Section~\ref{sec:approx-monotone}.

\begin{figure}[t]
    \centering
    \begin{subfigure}{.32\textwidth}
        \centering
        \begin{tikzpicture}[
                axis/.style={very thick, ->},
                curve/.style={thick},
                point/.style={fill=black,minimum size=0.02cm,circle},
                txt/.style={fill=none,draw=none},
                scale=.8,
                every node/.style={scale=.8},
            ]
            \draw[draw=none,fill=none] (-.7,-.8) rectangle (5.8,4.8);
            \draw[axis] (0,0)--(5.3,0); 
            \draw[axis] (0,0)--(0,4) node[midway,above,sloped] {$R_D(q)$}; 
            \filldraw[black] (0,0) circle (1pt) node[below]{$0$};
            \filldraw[black] (5,0) circle (1pt) node[below]{$1$};
            \filldraw[red] (4,0) circle (2pt) node[below=1pt,black] {$q=\frac{4}{5}$};
            %
            \draw[style=curve,domain=0:2] plot ({\x}, {\x*(5-\x)*2/5});
            \draw[style=curve,domain=3:5] plot ({\x}, {\x*(5-\x)*3/5});
            \draw[style=curve,domain=2:3] plot ({\x}, {\x*6/5});
            \filldraw[black] (2,12/5) circle (2pt);
            \filldraw[white,draw=black,thick] (3,18/5) circle (2pt);
            \filldraw[red] (4,12/5) circle (2pt);
            \draw[dashed,red,semithick] (0,0)--(4,12/5) node [midway,above=-2pt,sloped,black] {$\mbox{slope}=v$};
            \draw[dashed,red,semithick] (7/2,33/10)--(4,12/5)--(9/2,15/10) node [midway,above,sloped,black] {$\mbox{right derivative}=\virtual(v)$};
            \draw[dashed,red,semithick] (4,12/5)--(4,0);
        \end{tikzpicture}
        \caption{Non-point-mass, $v = \frac{1}{5}$}
        \label{fig:revenue-curve-continous}
    \end{subfigure}
    \begin{subfigure}{.32\textwidth}
        \centering
        \begin{tikzpicture}[
                axis/.style={very thick, ->},
                curve/.style={thick},
                point/.style={fill=black,minimum size=0.02cm,circle},
                txt/.style={fill=none,draw=none},
                scale=.8,
                every node/.style={scale=.8},
            ]
            \draw[draw=none,fill=none] (-.7,-.8) rectangle (5.8,4.8);
            \draw[axis] (0,0)--(5.3,0); 
            \draw[axis] (0,0)--(0,4) node[midway,above,sloped] {$R_D(q)$}; 
            \filldraw[black] (0,0) circle (1pt) node[below]{$0$};
            \filldraw[black] (5,0) circle (1pt) node[below]{$1$};
            \filldraw[black] (2.5,0) circle (0pt) node[below=1pt] {$q \in [\frac{2}{5}, \frac{3}{5})$};
            \draw[style=curve,domain=0:2] plot ({\x}, {\x*(5-\x)*2/5});
            \draw[style=curve,domain=3:5] plot ({\x}, {\x*(5-\x)*3/5});
            \draw[style=curve,red,semithick] (2,12/5)--(3,18/5) node [midway,above,sloped,black,align=center] {$\mbox{right derivative}$\\$=\virtual(v)$};
            \draw[dashed,red,semithick] (2,0)--(2,12/5);
            \draw[dashed,red,semithick] (3,0)--(3,18/5);
            \filldraw[red] (2,12/5) circle (2pt);
            \filldraw[red] (2,0) circle (2pt);
            \filldraw[white,draw=red,thick] (3,18/5) circle (2pt);
            \filldraw[white,draw=red,thick] (3,0) circle (2pt);
            \draw[dashed,red,semithick] (0,0)--(2,12/5) node [midway,below=-2pt,sloped,black] {$\mbox{slope}=v$};
        \end{tikzpicture}
        \caption{Point mass, $v=\frac{1}{2}$}
        \label{fig:revenue-curve-point-mass}
    \end{subfigure}
    \begin{subfigure}{.32\textwidth}
        \centering
        \begin{tikzpicture}[
                axis/.style={very thick, ->},
                curve/.style={thick},
                point/.style={fill=black,minimum size=0.02cm,circle},
                txt/.style={fill=none,draw=none},
                scale=.8,
                every node/.style={scale=.8},
            ]
            \draw[draw=none,fill=none] (-.7,-.8) rectangle (5.8,4.8);
            \draw[axis] (0,0)--(5.3,0); 
            \draw[axis] (0,0)--(0,4) node[midway,above,sloped] {$R_D(q)$}; 
            \filldraw[black] (0,0) circle (1pt) node[below]{$0$};
            \filldraw[black] (5,0) circle (1pt) node[below]{$1$};
            \filldraw[black] ({3-sqrt(3)/2},0) circle (0pt) node[below=1pt] {$q \in [\frac{3-\sqrt{3}}{5}, \frac{3}{5})$};
            \draw[style=curve,domain=0:(3-sqrt(3))] plot ({\x}, {\x*(5-\x)*2/5});
            \draw[style=curve,domain=(3-sqrt(3)):2,Gray] plot ({\x}, {\x*(5-\x)*2/5});
            \draw[style=curve,domain=2:3,Gray] plot ({\x}, {\x*6/5});
            \draw[style=curve,domain=3:5] plot ({\x}, {\x*(5-\x)*3/5});
            \filldraw[Gray] (2,12/5) circle (2pt);
            \draw[style=curve,thick,red] ({3-sqrt(3)},{(3+sqrt(3))*2/5})--(3,18/5) node [midway,above,sloped,black,align=center] {$\mbox{right derivative}$\\$=\bar{\virtual}(v)$}; 
            \draw[dashed,red,semithick] ({3-sqrt(3)},0)--({3-sqrt(3)},{(3+sqrt(3))*2/5});
            \draw[dashed,red,semithick] (3,0)--(3,18/5);
            \filldraw[fill=white,draw=red,thick] (3,18/5) circle (2pt);
            \filldraw[red] ({3-sqrt(3)},{(3+sqrt(3))*2/5}) circle (2pt);
            \filldraw[red] ({3-sqrt(3)},0) circle (2pt);
            \filldraw[white,draw=red,thick] (3,0) circle (2pt);
        \end{tikzpicture}
        \caption{Ironed Revenue Curve}
        \label{fig:ironed-revenue-curve}
    \end{subfigure}
    \caption{Revenue curves. This example corresponds to a value distribution that has a point mass of $\frac{1}{5}$ at value $\frac{1}{2}$, and otherwise is uniform over $[0, 1]$.}
    \label{fig:revenue-curve}
\end{figure}
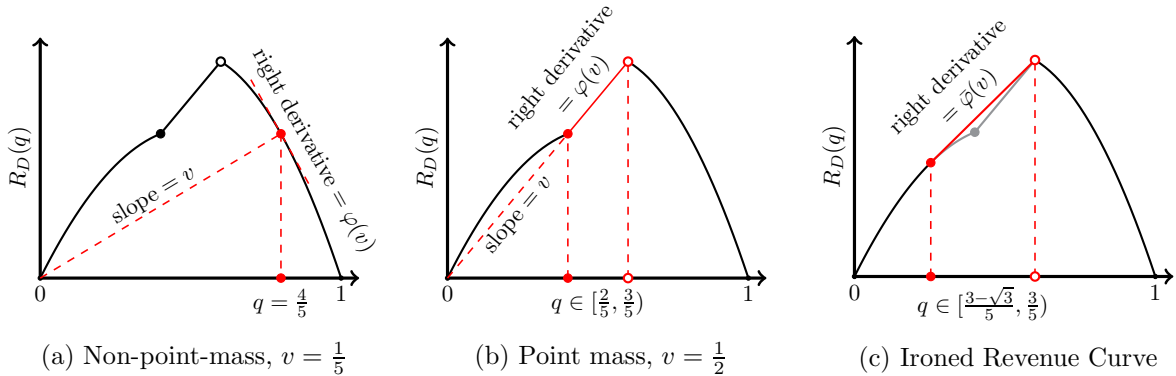

Finally, we formally define strong revenue monotonicity below.

\begin{definition}
    \label{def:strong-revenue-monoton}
    The single-parameter auctions with feasible allocations $\allocset$ satisfy \emph{strong revenue monotonicity} if for any product value distribution $\tilde{D}$ and any stochastically dominating product value distribution $D \succeq \tilde{D}$, Myerson's optimal auction $M_{\tilde{D}}$ w.r.t.\ $\tilde{D}$ gets weakly larger revenue from $D$ than from $\tilde{D}$, i.e.:
    \[
        M_{\tilde{D}}(D) \ge \OPT(\tilde{D})
        ~.
    \]
\end{definition}


\subsection{Set Systems and Matroids}

The simplified definition of single-parameter auctions in the introduction further assumes that the set of feasible allocations consists of some binary allocations, in which every $x_i \in \{0, 1\}$, and their convex combinations through randomized allocations.
We can describe the set of feasible allocations of such an auction by a \emph{set system}.
Let $\indset$ be the set of subsets of bidders that can be allocated to in any feasible allocation;
the notation $\indset$ comes from the concept of \emph{independent sets} in set systems.
For any $S \in \indset$, let $x^S$ be the corresponding allocation vector, i.e., $x^S_i = 1$ if $i \in S$ and $x^S_i = 0$ otherwise.
The set of feasible allocation can then be written as:
\[
    \Big\{ ~ \sum_{S \in \indset} \lambda^S x^S ~ : ~ \sum_{S \in \indset} \lambda^S = 1 \mbox{ and } \forall S \in \indset, \lambda^S \ge 0 ~ \Big\}
\]

A set system is a \emph{matroid} if it satisfies:
\begin{enumerate}
    \item[(M1)] The empty set is feasible, i.e., $\varnothing \in \indset$.
    \item[(M2)] The feasible sets are \emph{downward-closed}, i.e., if $S' \subseteq S$ and $S \in \indset$ then $S' \in \indset$.
    \item[(M3)] The feasible sets satisfy the \emph{exchange property}, i.e., for any feasible sets $S, S' \in \indset$ such that $|S'| < |S|$, there is a bidder $i \in S \setminus S'$ that can be allocated to on top of $S'$, i.e., $S' \cup \{i\} \in \indset$.
\end{enumerate}

Naturally it is a \emph{downward-closed} set system if it satisfies the first two properties.

\subsection{Empirical Distributions}

Let $\bD = D_1 \times D_2 \times \dots \times D_n$ be a product value distribution.
Given $N$ i.i.d.\ samples from $\bD$, the \emph{product empirical distribution} $\bE = E_1 \times E_2 \times \dots \times E_n$ is defined such that each dimension $E_i$ is the uniform distribution over the samples from $D_i$.
The next lemma follows from Bernstein's inequality and union bound.

\begin{lemma}[c.f., Lemma 5 of \citet{GuoHZ:STOC:2019}]
    \label{lem:empirical-cdf-bound}
    For any product distribution $\bD$ on $\R_+^n$, any positive integer $N$, and any $0 < \delta < 1$, consider the product empirical distribution $\bE$ from $N$ i.i.d.\ samples from $\bD$.
    Then, with probability at least $1-\delta$, for any bidder $1 \le i \le n$:
    \[
        \max_{v \in \R_+} \big\vert D_i(v) - E_i(v) \big\vert \le \sqrt{\frac{2 D_i(v) (1-D_i(v)) \ln \frac{2nN}{\delta}}{N}} + \frac{\ln \frac{2nN}{\delta}}{N}
        ~.
    \]
\end{lemma}

Further define the \emph{dominated product empirical distribution} $\btildeE = \tilde{E}_1 \times \tilde{E}_2 \times \dots \times \tilde{E}_n$ such that for any bidder $i$ and any value $0 \le v_i \le 1$:
\[
    \tilde{E}_i(v_i) = \min \Big\{ \, 1 ~,~ E_i(v_i) + \sqrt{\frac{2 E_i(v) (1-E_i(v)) \ln \frac{2nN}{\delta}}{N}} + \frac{4 \ln \frac{2nN}{\delta}}{N} \, \Big\}
    ~.
\]
The larger constant in the last term upper bounds the difference from having distinct terms inside the square root compared to the previous equation.

\begin{lemma}[c.f., Lemmas 6 and 7 of \citet{GuoHZ:STOC:2019}]
    \label{lem:dominated-empirical-cdf-bound}
    For any product distribution $\bD$ on $\R_+^n$, any positive integer $N$, and any $0 < \delta < 1$, consider the dominated product empirical distribution $\btildeE$ from $N$ i.i.d.\ samples from $\bD$.
    Then, with probability at least $1-\delta$, for any bidder $1 \le i \le n$ and any value $0 \le v_i \le 1$:
    \[
        D_i(v_i) \le \tilde{E}_i(v_i) \le  D_i(v_i) + \sqrt{\frac{8 D_i(v) (1-D_i(v)) \ln \frac{2nN}{\delta}}{N}} + \frac{7\ln \frac{2nN}{\delta}}{N}
        ~.
    \]
\end{lemma}

Motivated by the above bounds on the differences between the original distribution and empirical distributions and driven by our analysis, we say that two product value distributions $\bD$ and $\bE$ are \emph{$\epsilon$-close}, denoted as $\bD \approx_\epsilon \bE$, if for any bidder $i$ and any value $0 \le v_i \le 1$ we have:
\[
    \big\vert \, D_i(v_i) - E_i(v_i) \, \big\vert ~\le~ \sqrt{\min \Big\{ D_i(v_i) \big(1-D_i(v_i)\big) \,,\, E_i(v_i) \big(1-E_i(v_i)\big) \Big\} \cdot \frac{\epsilon^2}{4nk}} + \frac{\epsilon^2}{2nk}
    ~.
\]

The value distribution and dominated product empirical distribution are $\epsilon$-close when the number of samples meets the sample complexity bound in this paper.
The next lemma, which follows as a corollary of Lemma~\ref{lem:dominated-empirical-cdf-bound}, makes this precise.

\begin{lemma}
    \label{lem:dominated-empirical-cdf-bound-2}
    For any product distribution $\bD$ on $\R_+^n$, any positive integer $N$, and any $0 < \delta < 1$, consider the dominated product empirical distribution $\btildeE$ from $N$ i.i.d.\ samples from $\bP$.
    If:
    \[
        N \ge C \cdot \frac{nk}{\epsilon^2} \log \frac{nk}{\epsilon \delta}
        ~,
    \]
    for a sufficiently large constant $C$, then with probability at least $1-\delta$ we have both $\bD \succeq \btildeE$ and $\bD \approx_\epsilon \btildeE$.
\end{lemma}

Some readers may prefer a uniform notion of $\epsilon$-closeness.
Two product value distributions $D$ and $E$ are $\epsilon$-close uniformly, denote as $D \uniformclose_\epsilon E$, if for any bidder $i$ and any value $0 \le v_i \le 1$:
\[
    \big\vert ~ D_i(v_i) - E_i(v_i) ~ \big\vert \le \frac{\epsilon}{\sqrt{nk}}
    ~,
\]
where the right-hand-side is scaled by $\frac{1}{\sqrt{nk}}$ for a direct comparison with the nonuniform notion.
By definition, $D \approx_\epsilon E$ implies $D \uniformclose_\epsilon E$ but not the other way around.

\section{Strong Revenue (Non-)Monotonicity}
\label{sec:rev-monotone}

\subsection{Example of Non-monotonicity}
\label{sec:non-monotone}

Recall that \citet{DevanurHP:STOC:2016} proved the strong revenue monotonicity of single-parameter auctions in the matroid setting, i.e., when the set of feasible allocations is the convex hull of the basis of a matroid.
This subsection shows that the strong revenue monotonicity fails to generalize to more general single-parameter auctions (for some value distributions) even if we assume downward-closedness of the feasible allocations.

\paragraph{Example (Minimum Non-matroid).}
    Consider a $3$-bidder rank-$2$ auction as follows.
    We refer to the three bidders as $A$, $B$ and $C$.
    Let the set of feasible allocations be:
    \[
        \allocset = \Big\{ (x_A, x_B, x_C) \in [0,1]^3 : x_A + x_B \le 1 \mbox{ and } x_A + x_C \le 1 \Big\}
        ~.
    \]
    In other words, it is the convex hull of allocating exclusively to $A$, i.e., $x = (1, 0, 0)$, allocating to one or both of $B$ and $C$, i.e., $x = (0, 1, 1)$, $(0, 1, 0)$, or $(0, 0, 1)$, and not allocating anything, i.e., $x = (0, 0, 0)$.

\begin{theorem}
    \label{thm:non-monotone}
    The single-parameter auction whose feasible allocations are the above minimum non-matroid does not satisfy strong revenue monotonicity (for some value distributions).
\end{theorem}
\begin{proof}
    Consider two value distributions $\bD$ and $\btildeD$ as follows.
    The value of $A$ is deterministically $\frac{1}{2}$ in both $\bD$ and $\btildeD$, i.e., its value distributions are identically a point mass.
    By definition its virtual value is also deterministically $\frac{1}{2}$

    The values of $B$ and $C$ are independently and identically distributed in both distributions.
    In the smaller distribution $\btildeD$, the value of each is $1$ with probability $\epsilon = \frac{1}{10}$ and is $\epsilon$ otherwise.
    In the larger distribution $\bD$, the values are deterministically $1$.
    By definition, the revenue curves w.r.t.\ the smaller distribution $\btildeD$ are as given in Figure~\ref{fig:non-monotone-example-virtual-value} and the corresponding virtual values are:
    %
    \[
        \virtual_{\tilde{D}_B}(v) = \virtual_{\tilde{D}_C}(v) =
        \begin{cases}
            1 & \mbox{ if $v = 1$;} \\
            0 & \mbox{ if $\epsilon \le v < 1$;} \\
            -\infty & \mbox{ if $0 \le v < \epsilon$.}
        \end{cases}
    \]
    
    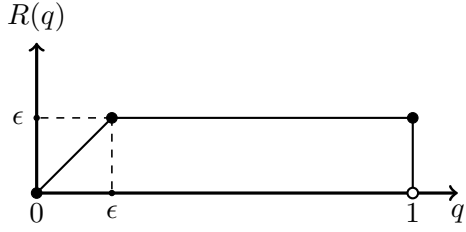
\begin{figure}[t]
        \centering
        \begin{tikzpicture}[
                axis/.style={very thick, ->},
                curve/.style={thick},
                point/.style={fill=black,minimum size=0.02cm,circle},
                txt/.style={fill=none,draw=none},
                scale=1.0,
                every node/.style={scale=1.0},
            ]
            \draw[axis] (0,0)--(5.6,0) node[at end,below] {$q$}; 
            \draw[axis] (0,0)--(0,2) node[at end,above] {$R(q)$}; 
            %
            %
            \draw[black,curve] (0,0) -- (1,1);
            \draw[black,curve] (1,1) -- (5,1);
            \draw[black,curve] (5,1) -- (5,0);
            \draw[dashed,black,semithick] (1,1) -- (1,0);
            \draw[dashed,black,semithick] (1,1) -- (0,1);
            \filldraw[draw=black,fill=black] (0,0) circle (2pt) node[below]{$0$};
            \filldraw[draw=black,thick,fill=white] (5,0) circle (2pt) node[below]{$1$};
            \filldraw[black] (1,1) circle (2pt);
            \filldraw[black] (5,1) circle (2pt);
            \filldraw[black] (1,0) circle (1pt) node[below=1pt] {$\epsilon$};
            \filldraw[black] (0,1) circle (1pt) node[left=1pt] {$\epsilon$};
        \end{tikzpicture}
        \caption{Revenue curve of $\tilde{D}_B$ and $\tilde{D}_C$ in the proof of Theorem~\ref{thm:non-monotone}}
        \label{fig:non-monotone-example-virtual-value}
    \end{figure}

    Therefore, Myerson's optimal auction $\myerson_{\btildeD}$ w.r.t.\ $\btildeD$ is as follows:
    \begin{itemize}
        \item If $v_B = v_C = 1$, allocate to both $B$ and $C$, and let each of them pay $\epsilon$; that is, $x = (0, 1, 1)$ and $p = (0, \epsilon, \epsilon)$.
        \item If $v_B = 1$, $v_C = \epsilon$, allocate to both $B$ and $C$, and let $B$ pay $1$ and let $C$ pay $\epsilon$; that is, $x = (0, 1, 1)$ and $p = (0, 1, \epsilon)$.
        \item If $v_B = \epsilon$, $v_C = 1$, allocate to both $B$ and $C$, and let $B$ pay $\epsilon$ and let $C$ pay $1$; that is, $x = (0, 1, 1)$ and $p = (0, \epsilon, 1)$.
        \item If $v_B = v_C = \epsilon$, allocate to $A$, let $A$ pay $\frac{1}{2}$; that is, $x = (1, 0, 0)$ and $p = (\frac{1}{2}, 0, 0)$.
    \end{itemize}

    The corresponding expected revenue equals:
    \[
        M_{\btildeD}(\btildeD) = (1-\epsilon)^2 \cdot \frac{1}{2} + 2 \epsilon(1-\epsilon) \cdot (1+\epsilon) + \epsilon^2 \cdot 2\epsilon > \frac{1}{2}
        ~.
    \]

    The above Myerson's optimal auction w.r.t.\ $\btildeD$ suffers from non-monotone payments:
    when only one of $B$ and $C$ has value $1$ the total payment is $1 + \epsilon$; when both of them have value $1$, however, the total payment is only $2\epsilon$.
    That is why in the bigger value distribution $\bD$ we let $B$ and $C$'s value distributions be identically a point mass at $1$.
    As a result, the expected revenue is only
    \[
        M_{\btildeD}(\bD) = 2\epsilon
        ~.
    \]

    By our choice of $\epsilon = \frac{1}{10}$ we have $M_{\btildeD}(\bD) < M_{\btildeD}(\btildeD) - \Omega(1)$.
\end{proof}

\begin{remark}
    In the above construction, $\epsilon$ can be arbitrarily small.
    Therefore, the multiplicative gap between the optimal revenue of a value distribution $\tilde{D}$ and the revenue of running Myerson's optimal auction w.r.t.\ $\tilde{D}$ on a stochastically dominating distribution $D$ can be arbitrarily large.
\end{remark}

By making $\lfloor \frac{k}{2} \rfloor$ copies of the minimum non-matroid and the value distributions in Theorem~\ref{thm:non-monotone}, we get that the decrease in revenue could be as large as $\Omega(k)$, which is at least a constant factor of the optimal revenue since the latter cannot exceed $k$.

\begin{corollary}
    \label{cor:non-monotone}
    For any positive integer $k$, there is a rank-$k$ downward-closed single-parameter auction, and two value distributions $\bD \succeq \btildeD$ such that:
    \[
        M_{\btildeD}(\bD) < M_{\btildeD}(\btildeD) - \Omega(k)
        ~.
    \]
\end{corollary}

Finally, we build on the above example to show that the matroid setting is the only case satisfying strong revenue monotonicity among all downward-closed set systems.%
\footnote{We need the assumption of downward-closedness because otherwise it could be a non-matroid but also effectively equivalent to a matroid, e.g., by adding a dummy bidder $i$ who must be allocated to, i.e., $x_i = 1$, to a matroid setting.}

\begin{theorem}
    \label{thm:strong-revenue-monotone-iff-matroid}
    If the set of feasible allocations of a single-parameter auction is a downward-closed set system, then it satisfies strong revenue monotonicity (for all value distributions) if and only the set system is a matroid.
\end{theorem}

\begin{proof}
    \citet{DevanurHP:STOC:2016} proved the direction from matroid to strong revenue monotonicity.
    It remains to prove that any downward-closed non-matroid auction does not satisfy strong revenue monotonicity.
    Consider any such auction and its set $\indset$ of feasible subsets of bidders who could be allocated to.
    Since it is downward-closed and is not a matroid, it must violate the exchange property (M3).
    We will build on this fact to find an embedded structure that resembles the aforementioned minimum non-matroid example.
    
    By the violation of exchange property, there are two feasible subsets of bidders $S, S' \in \indset$ such that $|S'| < |S|$, yet for any bidder $i \in S \setminus S'$, it is infeasible to allocate to $i$ on top of $S'$, i.e.:
    \[
        \forall i \in S \setminus S' : \quad S' \cup \{i\} \notin \indset
        ~.
    \]

    Note that the above implies $S' \not\subset S$.
    By removing bidders from $S$, we may assume WLOG that $|S| = |S'|+1$.
    Among all such pairs of $S$ and $S'$, consider the one with the maximum $|S \cap S'|$.
    Define $S_\cap = S \cap S'$, $S_A = S' \setminus S$, and $S_{BC} = S \setminus S'$.
    Further let $m = |S_\cap|$ and $\ell = |S_{BC}|$.
    Then, we have $|S| = m + \ell$ and $|S'| = m+\ell-1$.
    The above greedy choice of $S$ and $S'$ gives a useful property.

    \begin{lemma}
        \label{lem:crossing-sets}
        Suppose that a feasible subset of bidders $T$ satisfies that (1) $S_\cap \subseteq T$, (2) $T \cap S_A \ne \emptyset$, and (3) $T \cap S_{BC} \ne \emptyset$.
        We have $|T| \le m + \ell -2$.
    \end{lemma}

    It follows from a proof by contradiction.
    Suppose on the contrary that there exists such a $T$, and consider the one with the most elements.
    If $|T| = m + \ell - 1$, we may replace $S'$ by $T$ and increase $|S \cap S'|$.
    If $|T| \ge m + \ell$, we may replace $S$ by $T$ and increase $|S \cap S'|$.

    We next describe the value distributions $\bD$ and $\btildeD$.
    Let the values of all bidders in $S_\cap$ be $1$ deterministically in both distributions, and thus their virtual values are also deterministically $1$.
    Recall that $n$ is the number of bidders.
    The values of other bidders will be at most $\frac{1}{n}$ so that their total contribution to virtual welfare is less than that of a single bidder in $S_\cap$.
    Hence, Myerson's optimal auction always allocates to bidders in $S_\cap$ and gets $m = |S_\cap|$ total virtual values from them.

    Further let the values of all bidders not in $S \cup S'$ be $0$ deterministically so that they may be ignored in our discussion.

    Next, let $A$ be an arbitrary bidder in $S_A$.
    Let $B$ and $C$ be two arbitrary bidders in $S_{BC}$.
    Let the values of all bidders in $S_A \setminus \{A\}$ and $S_{BC} \setminus \{B, C\}$ be $\frac{1}{n}$ deterministically in both distributions.
    Let the value distributions of $A$, $B$, and $C$ be the same as the construction in Theorem~\ref{thm:non-monotone}, scaled by a $\frac{1}{n}$ factor so that they are at most $\frac{1}{n}$ as claimed earlier.

    Consider a feasible allocation that allocates to all bidders in $S_\cap$.
    If it is $S_\cap \cup S_A$, the virtual welfare equals:
    \[
        \underbrace{\vphantom{\Big|}m}_{\mbox{\small from $S_\cap$}} + \underbrace{\vphantom{\Big|}\frac{\ell-2}{n}}_{\mbox{\small from $S_A \setminus \{A\}$}} + \quad \underbrace{\vphantom{\Big|}\virtual_A(v_A)}_{\mbox{\small from $A$}}
        ~.
    \]

    If it is $S_\cap \cup S_{BC}$, the virtual welfare equals:
    \[
        \underbrace{\vphantom{\Big|}m}_{\mbox{\small from $S_\cap$}} + \underbrace{\vphantom{\Big|}\frac{\ell-2}{n}}_{\mbox{\small from $S_{BC} \setminus \{B, C\}$}} + \quad \underbrace{\vphantom{\Big|}\virtual_B(v_B) ~ + ~ \virtual_C(v_C)}_{\mbox{\small from $A$}}
        ~.
    \]

    If it involves bidders from both $S_A$ and $S_{BC}$, by Lemma~\ref{lem:crossing-sets} the virtual welfare is at most:
    \[
        m ~ + ~ \frac{\ell-2}{n}
        ~.
    \]
    
    Therefore, Myerson's optimal auction allocates to either $S_\cap \cup S_A$ or $S_\cap \cup S_{BC}$ at all time, and the contributions to the virtual welfare from bidders other than $A$, $B$, and $C$ are constant.
    The rest reduces to the aforementioned example and Theorem~\ref{thm:non-monotone}.
\end{proof}

\subsection{Approximate Strong Revenue Monotonicity}
\label{sec:approx-monotone}

In the above example, distribution $\bD$ not only stochastically dominates $\btildeD$ but is also much bigger.
We next show that if $\bD$ is instead $\epsilon$-close to $\btildeD$ then we recover an approximate version of strong revenue monotonicity.

\begin{theorem}
    \label{thm:approximate-revenue-monotone}
    For any $0 < \epsilon \le 1$, if product value distributions $\bD$ and $\btildeD$ satisfy $\bD \succeq \btildeD$ and $\bD \approx_\epsilon \btildeD$, then we have:
    \[
        M_{\btildeD}(\bD) \ge M_{\btildeD}(\btildeD) - \epsilon
        ~.
    \]
    If instead $\bD \uniformclose_\epsilon \btildeD$, then we have:
    \[
        M_{\btildeD}(\bD) \ge M_{\btildeD}(\btildeD) - \sqrt{\frac{n}{k}} \cdot \epsilon
        ~.
    \]
\end{theorem}

To prove Theorem~\ref{thm:approximate-revenue-monotone} we need to analyze $M_{\btildeD}(\bD)$.
Recall that Myerson's optimal auction $M_{\btildeD}$ w.r.t.\ $\btildeD$ chooses an allocation based on the virtual values w.r.t.\ $\btildeD$.
When it allocates to a bidder $i$ who has quantile $q_i$ and thus value $v_{D_i}(q_i)$, it expects a contribution of $\virtual_{\tilde{D}_i}\big(v_{D_i}(q_i)\big)$ to the expected revenue.
The actual contribution, however, depends on the virtual values w.r.t.\ $\bD$, and therefore is $\virtual_{D_i}\big(v_{D_i}(q_i)\big)$.
The next lemma upper bounds how much auction $M_{\btildeD}$ might have overestimated a bidder $i$'s contribution to the expected revenue because of the aforementioned mismatch.
Here recall that if a bidder $i$ gets allocation $1$ at some quantile $\theta_i$ then it also gets this allocation for any quantile less than $\theta_i$, and that $\int_0^{\theta_i} \virtual_{D_i}\big(v_{D_i}(q_i)\big) dq_i = R_{D_i}(\theta_i)$.

\begin{lemma}
    \label{lem:single-bidder-bound}
    Suppose that product value distributions $\bD$ and $\btildeD$ satisfy (1) $\bD \succeq \btildeD$, (2) $\bD \approx_\epsilon \btildeD$, and (3) $\btildeD$ is regular.
    Then for any bidder $i$ and any threshold quantile $0 \le \theta_i \le 1$ such that $\virtual_{\tilde{D}_i}\big(v_{D_i}(\theta_i)\big) \ge 0$, we have:
    \[
        \int_0^{\theta_i} \virtual_{\tilde{D}_i} \big( v_{D_i}(q) \big) dq \: \leq \: R_{D_i}(\theta_i) + \sqrt{\frac{\theta_i \epsilon^2}{4nk}} + \frac{\epsilon^2}{2nk}
        ~.
    \]
    If instead the second condition is replaced with $\bD \uniformclose_\epsilon \btildeD$ then we have a weaker bound:
    \[
        \int_0^{\theta_i} \virtual_{\tilde{D}_i} \big( v_{D_i}(q) \big) dq \: \leq \: R_{D_i}(\theta_i) + \frac{\epsilon}{\sqrt{nk}}
        ~.
    \]
\end{lemma}

\begin{proof}
    Define auxiliary distributions $P_i$ and $\tilde{P}_i$ such that for any value $0 \le v_i \le 1$:
    \[
        P_i(v_i) = \max \Big\{ D_i(v_i), 1-\theta_i \Big\}
        ~,\quad
        \tilde{P}_i(v_i) = \max \Big\{ \tilde{D}_i(v_i), 1-\theta_i \Big\}
    \]
    In other words, let them be the same as $D_i$ and $\tilde{D}_i$ in quantile interval $[0, \theta_i)$, which is the focal point of the lemma, but round the value $v_i$ down to $0$ if its quantile is in interval $[\theta_i, 1]$.

    We first prove that for any $0 \le v_i \le 1$:
    \begin{equation}
        \label{eqn:single-bidder-bound-cdf}
        \tilde{P}_i(v_i) - P_i(v_i) \le
        \begin{cases}
            \sqrt{\frac{\theta_i \epsilon^2}{4nk}} + \frac{\epsilon^2}{2nk} & \mbox{if $\bD \approx_\epsilon \btildeD$;} \\[1ex]
            ~ \frac{\epsilon}{\sqrt{nk}} & \mbox{if $\bD \uniformclose_\epsilon \btildeD$.}
        \end{cases}
    \end{equation}

    By $\bD \succeq \btildeD$, we have $\tilde{D}_i(v_i) \ge D_i(v_i)$.
    If further $\tilde{D}_i(v_i) \le 1 - \theta_i$, then the left-hand-side of Eqn.~\eqref{eqn:single-bidder-bound-cdf} equals zero because $\tilde{P}_i(v_i) = P_i(v_i) = 1 - \theta_i$.
    Next suppose that $\tilde{D}_i(v_i) > 1 - \theta_i$ and thus $\tilde{P}_i(v_i) = \tilde{D}_i(v_i)$.
    We have:
    \begin{align*}
        \tilde{P}_i(v_i) - P_i(v_i)
        &
        \le \tilde{D}_i(v_i) - D_i(v_i)
        \tag{$\tilde{P}_i(v_i) = \tilde{D}_i(v_i), P_i(v_i) \ge D_i(v_i)$}
        \\
        &
        \le \sqrt{\frac{(1 - \tilde{D}_i(v_i)) \epsilon^2}{4nk}} + \frac{\epsilon^2}{2nk}
        \tag{$\bD \approx_\epsilon \btildeD$} \\
        &
        \le \sqrt{\frac{\theta_i \epsilon^2}{4nk}} + \frac{\epsilon^2}{2nk}
        ~.
    \end{align*}

    If instead we have $\bD \uniformclose_\epsilon \btildeD$ then the right-hand-side is $\frac{\epsilon}{\sqrt{nk}}$ after the second inequality.

    Further define an auxiliary $\tilde{\virtual}_i$ such that for any value $0 \le v_i \le 1$:
    \[
        \tilde{\virtual}_i(v_i) = \max \big\{ \virtual_{\tilde{D}_i}(v_i), 0 \big\}
        ~.
    \]

    Since the lemma assumes that $\virtual_{\tilde{D}_i}\big(v_{D_i}(\theta_i)\big) \ge 0$, we have $\tilde{\virtual}_i\big(v_{D_i}(q)\big) = \virtual_{\tilde{D}_i} \big(v_{D_i}(q)\big)$ for any $0 \le q \le \theta_i$.
    Further recall that values with quantiles larger than $\theta_i$ in $D_i$ are rounded down to $0$ in $P_i$.
    Hence, the left-hand-side of the lemma is equal to:
    \begin{align*}
        \int_0^{\theta_i} \tilde{\virtual}_i \big( v_{D_i}(q) \big) dq
        &
        = \int_0^1 \tilde{\virtual}_i(v_i) d P_i(v_i)
        \\
        &
        = \tilde{\virtual}_i(1) - \int_0^1 P_i(v_i) d \tilde{\virtual}_i(v_i)
        ~.
        \tag{Integration by parts}
    \end{align*}
    
    Since $\tilde{\virtual}_i(v_i)$ is nondecreasing and is between $0$ and $1$, the integral above may be viewed as the expected value of $P_i(v_i)$ when $v_i$ is drawn from a distribution with CDF $\tilde{\virtual}_i(v_i)$.
    Changing it to the expected value of $\tilde{P}_i(v_i)$ w.r.t.\ the same distribution would lead to an additive difference of at most $\max_{0 \le v_i \le 1} \big( \tilde{P}_i(v_i) - P_i(v_i) \big)$, i.e.:
    \begin{align*}
        \int_0^{\theta_i} \tilde{\virtual}_i \big( v_{D_i}(q) \big) dq
        &
        \le \tilde{\virtual}_i(1) - \int_0^1 \tilde{P}_i(v_i) d \tilde{\virtual}_i(v_i) + \max_{0 \le v_i \le 1} \big( \tilde{P}_i(v_i) - P_i(v_i) \big)
        \\
        &
        \le \tilde{\virtual}_i(1) - \int_0^1 \tilde{P}_i(v_i) d \tilde{\virtual}_i(v_i) + \sqrt{\frac{\theta_i \epsilon^2}{4nk}} + \frac{\epsilon^2}{2nk}
        \tag{Eqn.~\eqref{eqn:single-bidder-bound-cdf}}
        \\
        &
        = \int_0^1 \tilde{\virtual}_i(v_i) d \tilde{P}_i(v_i) + \sqrt{\frac{\theta_i \epsilon^2}{4nk}} + \frac{\epsilon^2}{2nk}
        \tag{Integration by parts}
        ~.
    \end{align*}
    
    If instead we have $\bD \uniformclose_\epsilon \btildeD$ then the last two terms on the right would be replaced with $\frac{\epsilon}{\sqrt{nk}}$ after the second inequality.


    Finally, substituting $\tilde{\virtual}_i$ by its definition:
    \[
        \int_0^1 \tilde{\virtual}_i(v_i) d \tilde{P}_i(v_i) = \int_0^{\theta_i} \max \big\{ \virtual_{\tilde{D}_i}\big(v_{\tilde{D}_i}(q)\big), 0 \big\} dq 
        ~.
    \]

    It remains to prove that the right-hand-side above is at most $R_{D_i}(\theta_i)$.
    If we have $\virtual_{\tilde{D}_i}\big(v_{\tilde{D}_i}(\theta_i)\big) \ge 0$, the above is simply $R_{\tilde{D}_i}(\theta_i)$, which is at most $R_{D_i}(\theta_i)$ because $\bD \succeq \btildeD$.

    Otherwise, let $v_i^*$ be the monopoly price w.r.t.\ $\tilde{D}_i$, i.e., $v_i^* \in \arg\max_p\, p \cdot \Pr_{v \sim \tilde{D}_i} [v \ge p]$.
    Further let $q_i^* = \Pr_{u \sim \tilde{D}_i} [u \ge v_i^*]$.
    Recall that $\tilde{D}_i$ is regular.
    The right-hand-side above equals $v_i^* q_i^*$ because it integrates the derivative of a concave revenue curve $R_{\tilde{D}_i}(q)$ past its peak, rounding negative derivatives up to $0$.
    On the one hand, by $\virtual_{\tilde{D}_i}(v_{\tilde{D}_i}(\theta_i)) < 0$ we have that $q_i^* \le \theta_i$.
    On the other hand, by the lemma assumption that $\virtual_{\tilde{D}_i}(v_{D_i}(\theta_i)) \ge 0$, we get that $v_i^* \le v_{D_i}(\theta_i)$.
    Putting together, $q_i^* v_i^*$ is at most $\theta_i v_{D_i}(\theta_i) \le R_{D_i}(\theta_i)$ (it holds with equality if $v_{D_i}(\theta_i)$ is not a point mass).
\end{proof}

Theorem~\ref{thm:approximate-revenue-monotone} now follows by applying Lemma~\ref{lem:single-bidder-bound} to all bidders and by using the Cauchy-Schwarz Inequality to bound the squared-roots of the threshold quantiles $\theta_i$'s, as we shall explain next.

\begin{proof}[Proof of Theorem~\ref{thm:approximate-revenue-monotone}]
    Let $\btildex$ denote the allocation rule of $\myerson_{\btildeD}$, Myerson's optimal auction w.r.t.\ $\btildeD$.
    By definition $\btildex$ allocates according to the bidders' ironed virtual values.
    Further recall that we may interpret ironing as rounding each bidder $i$'s value down to the closest value on the convex hull of revenue curve $R_{\tilde{D}_i}$, as explained in the Section~\ref{sec:prelim}.
    Hence, $\btildex$ effectively allocates according to the virtual values w.r.t.\ the distributions of the rounded values.
    By redefining both $\bD$ and $\btildeD$ w.r.t.\ the rounded values, we may assume WLOG that $\btildeD$ is regular and the corresponding virtual values $\virtual_{\tilde{D}_i}$'s are nondecreasing and coincide with the ironed virtual values.
    
    For example, suppose that for some bidder $i$, $\tilde{D}_i$ has a point mass of $\nicefrac{1}{5}$ at value $\nicefrac{1}{2}$ and is otherwise uniform over $[0, 1]$, as in the example from Section~\ref{sec:prelim}, and $D_i$ is a uniform distribution over $[\nicefrac{1}{2}, 1]$.
    Then, we may instead consider a regular value distribution $\tilde{D}'_i$ with a point mass of $\nicefrac{\sqrt{3}}{5}$ at value $\nicefrac{1}{2}$ and has probability density $\nicefrac{4}{5}$ in value intervals $[0, \nicefrac{1}{2})$ and $(\nicefrac{(1+\sqrt{3})}{4}, 1]$.
    Further consider $D'_i$ with a point mass of $\nicefrac{(\sqrt{3}-1)}{2}$ at value $\nicefrac{1}{2}$, and has density $2$ in value interval $(\nicefrac{(1+\sqrt{3})}{4}, 1]$.
    Myerson's optimal auctions w.r.t.\ $\tilde{D}_i$ and $\tilde{D}'_i$ are the same.
    Further, its allocations and payments when applied to $D_i$ and $D'_i$ are identical for any given quantile of bidder $i$.
    

    Next consider allocation $\btildex$ when bidders' values are drawn from $\bD$.
    For any bidder $i$, any quantile profile $\bq_{-i}$ of the other bidders, and any allocation level $0 \le y \le 1$, define $\theta_i(y, \bq_{-i})$ as the largest quantile below which bidder $i$ gets allocation at least $y$ when the other bidders' values are $\bv_{\bD_{-i}}(\bq_{-i})$.
    Formally:
    \[
        \theta_i(y, \bq_{-i}) = \sup \Big\{ 0 \le q_i \le 1 : \tilde{x}_i \big( v_{D_i}(q_i), \bv_{\bD_{-i}}(\bq_{-i}) \big) \ge y \Big\}
        ~.
    \]

    On the one hand, we write the expected revenue on the left-hand-side of the theorem as:
    \begin{align}
        M_{\btildeD}(\bD)
        &
        = \sum_{i=1}^n \int_{[0,1]^n} \tilde{x}_i\big(\bv_{\bD}(\bq)\big) \virtual_{D_i}\big(v_{D_i}(q_i)\big) d\bq
        \notag \\
        &
        = \sum_{i=1}^n \int_{[0,1]^{n-1}} \int_0^1 \int_0^{\theta_i(y,\bq_{-i})} \virtual_{D_i}\big(v_{D_i}(q_i)\big) dq_i dy d\bq_{-i} \label{eqn:approx-rev-monotone-integration-by-part} \\
        &
        = \sum_{i=1}^n \int_{[0,1]^{n-1}} \int_0^1 R_{D_i}\big(\theta_i(y,\bq_{-i})\big) dy d\bq_{-i} 
        ~.
        \notag
    \end{align}

    Here, Eqn.~\eqref{eqn:approx-rev-monotone-integration-by-part} holds because the previous step may be interpreted as integrating $\virtual_{D_i}\big(v_{D_i}(q_i)\big)$ over the area below the curve of $\tilde{x}_i\big(x_{D_i}(q_i), x_{\bD_{-i}}(\bq_{-i})\big)$ over the quantile space of $q_i$, while its right-hand-side may be viewed as integrating over the area on the left of curve $\theta_i(y, \bq_{-i})$ which is the inverse of $\tilde{x}_i$.
    See Figure~\ref{fig:integration-by-part-illustration} for an illustration.

    \begin{figure}[t]
        \centering
        \begin{tikzpicture}[
                axis/.style={very thick, ->},
                curve/.style={thick},
                point/.style={fill=black,minimum size=0.02cm,circle},
                txt/.style={fill=none,draw=none},
                scale=.85,
                every node/.style={scale=.85},
            ]
            \begin{scope}
                \draw[axis] (0,0) -- (5.5,0); 
                \draw[axis] (0,0) -- (0,4.5); 
                \draw[thick,dotted] (5,0) -- (5,4);
                \draw[thick,dotted] (0,4) -- (5,4); 
                \filldraw[black] (0,0) circle (2pt) node[below]{$0$};
                \filldraw[black] (5,0) circle (2pt) node[below]{$1$};
                \filldraw[black] (0,4) circle (2pt) node[left]{$1$};
                \node[txt]  at (2.6,-0.4) {$q_i$};
                \node[txt]  at (-1, 2)  {$\tilde{x}_i\big(v_{\bD}(\bq)\big)$};
                \draw[style=curve,domain=0:5] plot ({\x}, {-(\x-2)*(\x-2)*(\x-2)/10+3});
                \foreach \x in {1,...,9}
                    \draw[dashed,<->,semithick] ({\x/2},0)--({\x/2},{-(\x/2-2)*(\x/2-2)*(\x/2-2)/10+3});
            \end{scope}
            \begin{scope}[xshift=8cm]
                \draw[axis] (0,0) -- (5.5,0); 
                \draw[axis] (0,0) -- (0,4.5); 
                \draw[thick,dotted] (5,0) -- (5,4);
                \draw[thick,dotted] (0,4) -- (5,4); 
                \filldraw[black] (0,0) circle (2pt) node[below]{$0$};
                \filldraw[black] (5,0) circle (2pt) node[below]{$1$};
                \filldraw[black] (0,4) circle (2pt) node[left]{$1$};
                \node[txt]  at (2.6,-0.4)  {$\theta_i(y, \bq_{-i})$};
                \node[txt]  at (-.4, 2)   {$y$};
                \draw[style=curve,domain=0:5] plot ({\x}, {-(\x-2)*(\x-2)*(\x-2)/10+3});
                \foreach \y in {1,...,7}
                    \draw[dashed,<->,semithick] (0,{\y/2})--({2+sign(3-\y/2)*pow(abs(3-\y/2)*10,1/3)},{\y/2});
            \end{scope}
            \node[fill=none,draw=none] at (6.5,2) {\LARGE $\Leftrightarrow$};
        \end{tikzpicture}
        \caption{Illustrative picture for Equation~\eqref{eqn:approx-rev-monotone-integration-by-part}}
        \label{fig:integration-by-part-illustration}
    \end{figure}
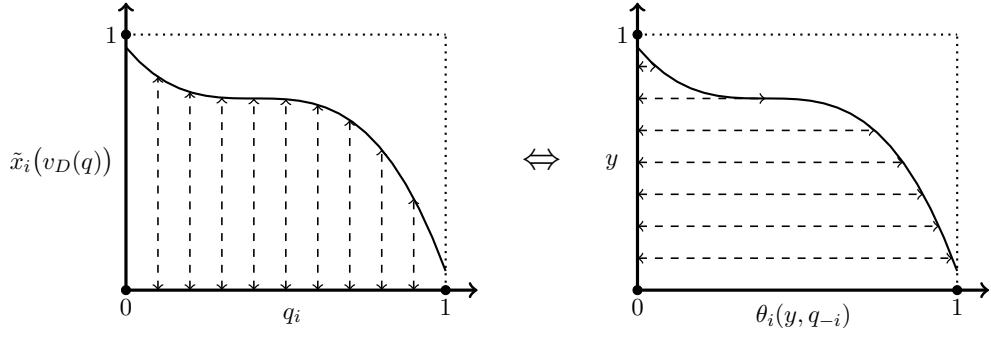

    On the other hand, we bound the expected revenue on right-hand-side by:
    \begin{align*}
        M_{\btildeD}(\btildeD)
        &
        = \int_{[0,1]^n} \sum_{i=1}^n \tilde{x}_i\big( \bv_{\btildeD}(\bq) \big) \virtual_{\tilde{D}_i}\big(v_{\tilde{D}_i}(q_i)\big) d\bq
        \\
        &
        \le \int_{[0,1]^n} \sum_{i=1}^n \tilde{x}_i\big( \bv_{\btildeD}(\bq) \big) \virtual_{\tilde{D}_i}\big(v_{D_i}(q_i)\big) d\bq
        \tag{$\bD \succeq \btildeD$} \\
        &
        \le \int_{[0,1]^n} \sum_{i=1}^n \tilde{x}_i\big( \bv_{\bD}(\bq) \big) \virtual_{\tilde{D}_i}\big(v_{D_i}(q_i)\big) d\bq
        \tag{$\btildex$ maximizes virtual welfare}
        \\
        &
        = \sum_{i=1}^n \int_{[0,1]^{n-1}} \int_0^1 \int_0^{\theta_i(y,\bq_{-i})} \virtual_{\tilde{D}_i}\big(v_{D_i}(q_i)\big) dq_i dy d\bq_{-i}
        ~.
        \tag{Same as Eqn.~\eqref{eqn:approx-rev-monotone-integration-by-part}}
    \end{align*}

    By Lemma~\ref{lem:single-bidder-bound}:
    \[
        \int_0^{\theta_i(y,\bq_{-i})} \virtual_{\tilde{D}_i}\big(v_{D_i}(q_i)\big) dq_i - R_{D_i}\big(\theta_i(y,\bq_{-i})\big) \le
        \begin{cases}
            \sqrt{\frac{\theta_i(y,\bq_{-i}) \epsilon^2}{4nk}} + \frac{\epsilon^2}{2nk} & \mbox{if $\bD \approx_\epsilon \btildeD$;} \\[1ex]
            \frac{\epsilon}{\sqrt{nk}} & \mbox{if $\bD \uniformclose_\epsilon \btildeD$.}
        \end{cases}
    \]

    Therefore, summing over $i$ and integrating over $q_{-i}$ and $y$ we have:
    \[
        M_{\btildeD}(\btildeD) - M_{\btildeD}(\bD)
        ~\le~ 
        \begin{cases}
            \sum_{i=1}^n \int_{[0,1]^{n-1}} \int_0^1 \Big( \sqrt{\frac{\theta_i(y,\bq_{-i}) \epsilon^2}{4nk}} + \frac{\epsilon^2}{2nk} \Big) dy d\bq_{-i} & \mbox{if $\bD \approx_\epsilon \btildeD$;} \\[1ex]
            \sqrt{\frac{n}{k}} \cdot \epsilon & \mbox{if $\bD \uniformclose_\epsilon \btildeD$.}
        \end{cases}
    \]

    It remains to finish proving the $D \approx_\epsilon \btildeD$ case. By Cauchy-Schwarz:
    \begin{align*}
        \sum_{i=1}^n \int_{[0,1]^{n-1}} \int_0^1 \sqrt{\theta_i(y,\bq_{-i})} dy d\bq_{-i}
        &
        \le \bigg(n \sum_{i=1}^n \int_{[0,1]^{n-1}} \int_0^1 \theta_i(y,\bq_{-i}) dy d\bq_{-i}\bigg)^{\frac{1}{2}}
        \\
        &
        = \bigg(n \sum_{i=1}^n \int_{[0,1]^n} \tilde{x}_i\big(\bv_{\bD}(\bq)\big) d\bq\bigg)^{\frac{1}{2}}
        \tag{Same as Eqn.~\eqref{eqn:approx-rev-monotone-integration-by-part}}
        \\[1ex]
        &
        \le \sqrt{nk} 
        ~.
        \tag{Rank $k$}
    \end{align*}

    Further:
    \[
        \sum_{i=1}^n \int_{[0,1]^{n-1}} \int_0^1 \frac{\epsilon^2}{2nk} dy d\bq_{-i} = \frac{\epsilon^2}{2k} \le \frac{\epsilon}{2}
        ~.
    \]

    Combining the bounds proves the theorem.
\end{proof}

\subsection{Strong Revenue Lipschitzness}

We remark that it is possible to remove the stochastic dominance condition from Theorem~\ref{thm:approximate-revenue-monotone} but we would need the two distributions to be closer than in that theorem.
We refer to this property as \emph{strong revenue Lipschitzness} of single-parameter auctions;
the next two theorems make it precise.

\begin{theorem}
    \label{thm:revenue-lipschitz}
    For some sufficiently small constant $c > 0$, and for any $0 < \epsilon \le 1$ if product value distributions $\bD$ and $\btildeD$ satisfy $\bD \approx_{\epsilon'} \btildeD$ for:
    \[
        \epsilon' = c \cdot \frac{\epsilon}{\sqrt{k \log \frac{nk}{\epsilon}}}
        ~,
    \]
    then we have:
    \[
        M_{\btildeD}(\bD) \ge M_{\btildeD}(\btildeD) - \epsilon
        ~.
    \]
\end{theorem}

\begin{proof}
    Define an auxiliary product value distribution $\bhatD$ such that for any bidder $i$ and any value $0 \le v_i \le 1$:
    \begin{align*}
        \hat{D}_i(v_i)
        &
        = D_i(v_i) + \sqrt{\frac{D_i(v_i) (1-D_i(v_i)) (\epsilon')^2}{4nk}} + \frac{2(\epsilon')^2}{nk} \\
        &
        = D_i(v_i) + \sqrt{\frac{D_i(v_i) (1-D_i(v_i)) c^2 \epsilon^2}{4nk^2 \log \frac{nk}{\epsilon}}} + \frac{2c^2 \epsilon^2}{nk^2 \log \frac{nk}{\epsilon}} 
    \end{align*}

    By this construction we have $\bhatD \approx_{O(\epsilon')} \btildeD$ which for a sufficiently small $c$ implies:
    \[
        \bhatD \approx_{\frac{\epsilon}{2}} \btildeD
        ~.
    \]    

    By $\bD \approx_{\epsilon'} \bhatD$ and by a calculation similar to Lemma~\ref{lem:dominated-empirical-cdf-bound}, we also get that:
    \[
        \bhatD \succeq \btildeD
        ~.
    \]

    Therefore by Theorem~\ref{thm:approximate-revenue-monotone} we have:
    \[
        M_{\btildeD}(\bhatD) \ge M_{\btildeD}(\btildeD) - \frac{\epsilon}{2}
    \]

    \begin{lemma}[c.f., Lemmas 1 and 11 of \citet{GuoHTZ:COLT:2021}]
        For $\bD$ and $\bhatD$ constructed above, and for any function $f : [0,1]^n \to [0,1]$ we have:
        \[
            \Big\vert ~ f(\bD) - f(\bhatD) ~ \Big\vert \le \frac{\epsilon}{2k}
            ~.
        \]        
    \end{lemma}

    Letting $f = \frac{1}{k} M_{\btildeD}$ gives $M_{\btildeD}(\bD) \ge M_{\btildeD}(\bhatD)  - \frac{\epsilon}{2}$.
    Combining it with the previous inequality proves the theorem.
\end{proof}



\begin{theorem}
    \label{thm:revenue-lipschitz-lower-bound}
    For any $n \ge 1$, any $k \le n$, and any $0 \le \epsilon \le 1$, there is a single-parameter auction with $n$ bidders and rank $k$, and two product value distribution $\bD \approx_{\epsilon} \btildeD$ such that:
    \[
        \OPT(\bD) \le \OPT(\btildeD) - \frac{\epsilon \sqrt{k}}{8}
        ~,
    \]
    and as a corollary:
    \[
        M_{\btildeD}(\bD) \le M_{\btildeD}(\btildeD) - \frac{\epsilon \sqrt{k}}{8}
        ~.
    \]
\end{theorem}

Theorem~\ref{thm:revenue-lipschitz-lower-bound} implies two conceptual messages when we compare its bound with those in Theorems~\ref{thm:approximate-revenue-monotone} and \ref{thm:revenue-lipschitz}.
First, the bound of strong revenue Lipschitzness given by Theorem~\ref{thm:revenue-lipschitz} is tight up to a logarithmic factor.
Second and more relevant to the theme of this paper, the closeness of value distributions on its own (strong revenue Lipschitzness) is strictly weaker than its combination with stochastic dominance (strong revenue monotonicity).
Therefore, if one has to use an inaccurate prior to design the auction, it is more robust to employ an underestimation.

\begin{proof}
    Consider an $n$-bidder rank-$k$ auction as follows.
    Let the set of feasible allocations be:
    \[
        \allocset = \Big\{ \big(0, 0, \dots, 0\big), \big(\,\frac{k}{n}, \frac{k}{n}, \dots, \frac{k}{n}\,\big) \Big\}
        ~.
    \]
    In other words, we allocate either to none of the bidders, or to all of them each by a $\frac{k}{n}$ amount.
    We next define the product value distributions $\bD$ and $\btildeD$.
    The bidders' values are binary, either $0$ or $1$, and are independently and identically distributed.
    In the first distribution $\bD$, the value of each bidder is $0$ with probability $\frac{1}{2n}$ and is $1$ with probability $1-\frac{1}{2n}$.
    In second distribution $\btildeD$, the value of each bidder is $0$ with a slightly smaller probability $\frac{1}{2n} -  \frac{\epsilon}{(4n\sqrt{k})}$, and is $1$ otherwise.
    By our construction it is easy to verify that $\bD \approx_\epsilon \btildeD$.
    
    For both distributions a bidder's virtual value is $1$ when its value is $1$, and is smaller than $-(n-1)$ when its value is $0$.
    Therefore, Myerson's optimal auctions w.r.t.\ these two distributions are the same:
    If all bidders have value $1$, allocate to all bidders each by a $\frac{k}{n}$ amount, and each bidder pays $\frac{k}{n}$.
    Otherwise, allocate to none of them, and the bidders pay nothing.

    
    
    Therefore, the optimal revenue w.r.t.\ the two distributions are:
    \[
        \OPT(D)
        =
        \underbrace{\vphantom{\bigg|} \frac{k}{n}}_{\mbox{\tiny payment per bidder}}
        \cdot
        \underbrace{\vphantom{\bigg|} n}_{\mbox{\tiny number of bidders}}
        \cdot
        \underbrace{\vphantom{\bigg|} \Big(1 - \frac{1}{2n}\Big)^n}_{\mbox{\tiny probability that all values are $1$}}
        = k \Big(1 - \frac{1}{2n}\Big)^n
        ~,
    \]
    and similarly:
    \begin{align*}
        \OPT(\btildeD) 
        &
        = \frac{k}{n} \cdot n \cdot \Big( 1 - \frac{1}{2n} + \frac{\epsilon}{4n\sqrt{k}} \Big)^n \\
        & = k \Big(1-\frac{1}{2n}\Big)^n \Big(1+\frac{\epsilon}{2(2n-1) \sqrt{k}} \Big)^n
        ~.
    \end{align*}
    
    The difference is therefore:
    \begin{align*}
        k \Big(1-\frac{1}{2n}\Big)^n \Big( \Big(1+\frac{\epsilon}{2(2n-1) \sqrt{k}} \Big)^n - 1 \Big)
        & 
        \ge k \cdot \frac{1}{2} \cdot \frac{\epsilon n}{2(2n-1) \sqrt{k}} && \mbox{\big($(1+x)^n \ge 1+nx$ for $x > -1$ \big)} \\
        &
        \ge \frac{\epsilon \sqrt{k}}{8} 
        ~.
        && \mbox{\big($\nicefrac{n}{(2n-1)} \ge \nicefrac{1}{2}$\big)}
    \end{align*}
\end{proof}

\section{Applications in Sample Complexity}
\label{sec:sample-complexity}

In practice auctioneers do not have accurate knowledge of the bidders' value distributions.
Instead they often have various kinds of data such as the bidders' bids in previous auctions.
\citet{ColeR:STOC:2014} introduced the sample complexity model of single-parameter auction design, in which the auctioneer can only access the value distributions $\bD$ through i.i.d.\ samples from it.
They asked how many samples are sufficient for learning an auction that is an (additive) $\epsilon$-approximation.
\citet{GuoHZ:STOC:2019} used the strong revenue monotonicity of single-parameter auctions in the matroid setting to obtain sample complexity bounds that are tight up to logarithmic factors.
This section shows that the approximate version of strong revenue monotonicity in the previous section is also sufficient for getting the same bound, but more generally in the downward-closed setting.
For arbitrary single-parameter auctions, we resort to the strong revenue Lipchitzness and get a sample complexity upper bound with an additional $\tilde{O}(k)$ factor.

\begin{theorem}[Downward-closed Feasibility, Unit-demand]
    \label{thm:sample-complexity}
    Consider an arbitrary downward-closed single-parameter auction with $n$ bidders and rank $k$.
    Suppose that we have $N$ i.i.d.\ samples where $N$ is at least:
    \[
        C \cdot \frac{nk}{\epsilon^2} \log \frac{nk}{\epsilon\delta}
    \]
    for a sufficiently large constant $C$.
    Then Myerson's optimal auction w.r.t.\ the dominated product empirical distribution $\btildeE$ is an $\epsilon$-approximation with probability at least $1-\delta$.
\end{theorem}

We first establish that the optimal revenue satisfies a Lipchitzness-style property.

\begin{lemma}[Downward-closed, Lipschitzness of Optimal Revenue]
    \label{lem:lipschitz-optimal-revenue}
    Consider an arbitrary downward-closed single-parameter auction.
    If we have both $\bD \succeq \btildeD$ and $\bD \approx_\epsilon \btildeD$, then:
    \[
        \OPT(\btildeD) \ge \OPT(\bD) - \epsilon
        ~.
    \]
\end{lemma}

A similar property was shown implicitly in the analysis of \citet{GuoHZ:STOC:2019} indirectly using an information theoretic argument and revenue monotonicity.
By comparison, our argument is more direct and reduces the logarithmic factors.

Readers may also notice that a weaker version of the lemma follows as a corollary of Theorem~\ref{thm:revenue-lipschitz} and holds more generally for arbitrary single-parameter auctions, although we would need the two distributions to be closer than the stated lemma assumption.

Next we first explain how our sample complexity bound (Theorem~\ref{thm:sample-complexity}) follows from a simple combination of the Lipschitzness of optimal revenue (Lemma~\ref{lem:lipschitz-optimal-revenue}) and approximate strong revenue monotonicity (Theorem~\ref{thm:approximate-revenue-monotone}).
Then we prove Lemma~\ref{lem:lipschitz-optimal-revenue}.

\begin{proof}[Proof of Theorem~\ref{thm:sample-complexity}]
    With a sufficiently large $N$ as assumed in the theorem, with probability at least $1-\delta$ we have that $\bD \succeq \btildeE$ and $\bD \approx_{\frac{\epsilon}{2}} \btildeE$.
    Therefore, the theorem follows by:
    \begin{align*}
        M_{\btildeE}(\bD)
        &
        \ge M_{\btildeE}(\btildeE) - \frac{\epsilon}{2}
        \tag{Theorem~\ref{thm:approximate-revenue-monotone}}
        \\
        &
        = \OPT(\btildeE) - \frac{\epsilon}{2} \\
        &
        \ge \OPT(\bD) - \epsilon
        ~.
        \tag{Lemma~\ref{lem:lipschitz-optimal-revenue}}
    \end{align*}
\end{proof}

\begin{proof}[Proof of Lemma~\ref{lem:lipschitz-optimal-revenue}]
    We will prove the lemma by constructing an auction whose expected revenue on value distribution $\btildeD$ is at least $\OPT(\bD) - \epsilon$.
    For ease of notations, for any bidder $i$ and any quantile $q_i$ we write $v_i(q_i)$ and $\tilde{v}_i(q_i)$ for $v_{D_i}(q_i)$ and $v_{\tilde{D}_i}(q_i)$ respectively.
    Let $\bx^*$ be the allocation of $\myerson_{\bD}$, Myerson's optimal auction w.r.t.\ $\bD$.
    We construct an allocation $\btildex$ such that for any bidder $i$ and any quantile vector $\bq$ and the corresponding value vector $\btildev(\bq)$ w.r.t.\ value distribution $\btildeD$, bidder $i$'s allocation is:
    \[
        \tilde{x}_i \big( \btildev(\bq) \big) = x^*_i \big( 
        \tilde{v}_i(q_i), \bv_{-i}(\bq_{-i}) \big)
        ~.
    \]
    That is, for every bidder $i$ we map the other bidders' values to what they would have been, if their value distributions were $\bD_{-i}$ instead of $\btildeD_{-i}$, and if their quantiles stayed the same.
    Then we call $x^*_i$ to decide bidder $i$'s allocation.
    The intuition behind this construction is best explained under the special case when allocation $\bx^*$ takes binary values in all coordinates.
    For any bidder $i$ and any realization of the other bidders' \emph{quantiles}, bidder $i$'s threshold value above which it gets allocation $1$ are the same in two scenarios: 
    (1) when the allocation rule is $x^*_i$ and the others' values are drawn from $\bD_{-i}$, and (2) when the allocation rule is $\tilde{x}_i$ and the others' values are drawn from $\btildeD_{-i}$.
    Hence, bidder $i$'s payments when being allocated are the same in these two scenarios.
    The only revenue loss under the latter scenario comes from the lower probability of reaching this threshold, which can be bounded by the assumption of $\bD \approx_{\frac{\epsilon}{2}} \btildeD$.
    
    We first argue that allocation rule $\btildex$ is feasible and monotone.
    For any bidder $i$, by the monotonicity of $x^*_i$ in bidder $i$'s value we have:
    \[
        \tilde{x}_i \big( \btildev(\bq) \big) \le x^*_i \big( \bv(\bq) \big)
        ~.
    \]
    Therefore, the feasibility of $\btildex$ follows from the feasibility of $\bx^*$ and that the feasible allocations are downward-closed.
    The monotonicity of $\tilde{x}_i$ in bidder $i$'s value follows by the monotonicity of $x^*_i$.
    
    By Myerson's Lemma, there is a payment rule which together with $\btildex$ form a truthful auction.
    Denote it by $\tilde{A}$.
    Write the virtual values of $D_i$ and $\tilde{D}_i$ by $\virtual_i$ and $\tilde{\virtual}_i$ respectively.
    We have:
    \begin{align*}
        \OPT(\bD)
        &
        = \sum_{i=1}^n \int_{[0,1]^n} x^*_i \big( \bv(\bq) \big) \virtual_i\big(v_i(q_i)\big) ~d\bq
        \\
        &
        = \sum_{i=1}^n \int_{[0,1]^{n-1}} \underbrace{\int_0^1 x^*_i \big( v_i(q_i), \bv_{-i}(\bq_{-i}) \big) \virtual_i\big(v_i(q_i)\big) ~dq_i}_{(a)} ~d\bq_{-i}
        ~,
    \end{align*}
    and:
    \begin{align*}
        \tilde{A}(\btildeD)
        &
        = \sum_{i=1}^n \int_{[0,1]^n} \tilde{x}_i \big( \btildev(\bq) \big) \tilde{\virtual}_i\big(v_i(q_i)\big)  ~d\bq
        \\
        &
        = \sum_{i=1}^n \int_{[0,1]^{n-1}} \underbrace{\int_0^1 x^*_i \big( \tilde{v}_i(q_i), \bv_{-i}(\bq_{-i}) \big) \tilde{\virtual}_i\big(\tilde{v}_i(q_i)\big) ~dq_i}_{(b)} ~d\bq_{-i}
        ~.
    \end{align*}

    Next for any bidder $i$ and any quantiles $\bq_{-i}$ of the other bidders, we will bound the difference between (a) and (b).
    For any $0 \le y \le 1$, let $v^*_i(y, \bq_{-i})$ be the minimum $v_i$ for which we have $x^*_i \big( v_i, \bv_{-i}(\bq_{-i}) \big) \ge y$.
    We have:
    \begin{align*}
        (a)
        &
        = \int_0^1 \int_{v^*_i(y, \bq_{-i})}^1 \virtual_i(v_i) ~dD_i(v_i) ~dy
        = \int_0^1 v^*_i(y, \bq_{-i}) \Pr_{v_i \sim D_i} \big[ v_i \ge v^*_i(y, \bq_{-i}) \big] ~dy
        ~,
        \\
        (b)
        &
        = \int_0^1 \int_{v^*_i(y, \bq_{-i})}^1 \tilde{\virtual}_i(v_i) ~d\tilde{D}_i(v_i) ~dy
        = \int_0^1 v^*_i(y, \bq_{-i}) \Pr_{v_i \sim \tilde{D}_i} \big[ v_i \ge v^*_i(y, \bq_{-i}) \big] ~dy
    \end{align*}
    
    The above equations hold for reasons similar to those behind Eqn.~\eqref{eqn:approx-rev-monotone-integration-by-part} and its illustration in Figure~\ref{fig:integration-by-part-illustration}, except that we use the value space as the $x$-coordinate instead of the quantile space.
    This choice simplifies the notations in the current argument, as the two equations correspond to allocations $x^*$ and $\tilde{x}$, which by design are related to each other through common threshold values.
    

    Further by $\bD \approx_\epsilon \btildeD$:
    \[
        \Pr_{v_i \sim \tilde{D}_i} \big[ v_i \ge v^*_i(y, \bq_{-i}) \big]
        \ge 
        \Pr_{v_i \sim D_i} \big[ v_i \ge v^*_i(y, \bq_{-i}) \big] - \sqrt{\frac{\epsilon^2}{4nk} \Pr_{v_i \sim D_i} \big[ v_i \ge v^*_i(y, \bq_{-i}) \big]} - \frac{\epsilon^2}{2nk}
        ~.
    \]


    Therefore:
    \[
        \OPT(\bD) - \tilde{A}(\btildeD) \le \sum_{i=1}^n \int_{[0,1]^{n-1}} \int_0^1 \Big( \sqrt{\frac{\epsilon^2}{4nk} \Pr_{v_i \sim D_i} \big[ v_i \ge v^*_i(y, \bq_{-i}) \big]} + \frac{\epsilon^2}{2nk} \Big) ~dy ~d\bq_{-i}
        ~.
    \]

    The lemma now follows by:
    \begin{align*}
        & \sum_{i=1}^n \int_{[0,1]^{n-1}} \int_0^1 \sqrt{\Pr_{v_i \sim D_i} \big[ v_i \ge v^*_i(y, \bq_{-i}) \big]} ~dy ~d\bq_{-i}
        \\
        & \qquad
        \le \bigg( n \sum_{i=1}^n \int_{[0,1]^{n-1}} \int_0^1 \Pr_{v_i \sim D_i} \big[ v_i \ge v^*_i(y, \bq_{-i}) \big] ~dy ~d\bq_{-i} \bigg)^{\frac{1}{2}}
        \tag{Cauchy-Schwarz}
        \\
        & \qquad
        = \bigg( n \sum_{i=1}^n \E_{\bv \sim \bD} x_i^*(\bv) \bigg)^{\frac{1}{2}}
        \\[1ex]
        & \qquad
        \le \sqrt{nk}
        ~,
        \tag{Rank $k$}
    \end{align*}
    and:
    \[
        \sum_{i=1}^n \int_{[0,1]^{n-1}} \int_0^1 \frac{\epsilon^2}{2nk} ~dy ~d\bq_{-i}
        = \frac{\epsilon^2}{2k} \le \frac{\epsilon}{2}
        ~.
    \]
\end{proof}

\begin{theorem}[General Feasibility, Unit-demand]
    \label{thm:sample-complexity-general}
    Consider an arbitrary single-parameter auction with $n$ bidders and rank $k$.
    Suppose that we have $N$ i.i.d.\ samples where $N$ is at least:
    \[
        C \cdot \frac{nk^2}{\epsilon^2} \log \frac{nk}{\epsilon} \log \frac{nk}{\epsilon\delta}
    \]
    for a sufficiently large constant $C$.
    Then Myerson's optimal auction w.r.t.\ the dominated product empirical distribution $\btildeE$ is an $\epsilon$-approximation with probability at least $1-\delta$.
\end{theorem}

\begin{proof}
    We will prove a $2\epsilon$-approximation with the understanding that halving the approximation factor does not change the sample complexity asymptotically.
    With the stated lower bound of $N$, we get that $\bD \approx_{\epsilon'} \btildeE$ for:
    \[
        \epsilon' = c \cdot \frac{\epsilon}{\sqrt{k \log \frac{nk}{\epsilon}}}
        ~.
    \]
    
    Therefore we have:
    \begin{align*}
        M_{\btildeE}(\bD)
        &
        \ge M_{\btildeE}(\btildeE) - \epsilon
        && \mbox{($\bD \approx_{\epsilon'} \btildeE$ and Theorem~\ref{thm:revenue-lipschitz})} \\
        &
        \ge M_{\bD}(\btildeE) - \epsilon
        && \mbox{(optimality of $M_{\btildeE}$)} \\
        &
        \ge M_{\bD}(\bD) - 2\epsilon
        ~.
        && \mbox{($\bD \approx_{\epsilon'} \btildeE$ and Theorem~\ref{thm:revenue-lipschitz})}
    \end{align*}
\end{proof}

We next complement Theorem~\ref{thm:sample-complexity-general} with an almost matching lower bound, demonstrating that the sample compelxity bound therein is tight up to logarithmic factors.
This lower bound also shows a separation between the sample complexity of downward-closed and non-downward-closed single-parameter auctions.

\begin{theorem}
    \label{thm:sample-complexity-general-lower-bound}
    There is a constant $c$ such that for any number of bidders $n$, any rank $k \le n$, and any $0 < \epsilon \le \frac{1}{100}$, there is a set of feasible allocations for which any algorithm needs at least:
    \[
        c \cdot \frac{n k^2}{\epsilon^2}
    \]
    samples to learn an auction with an $\epsilon$-approximation in expectation, over the randomness of the algorithm, the samples, and the bidders' valuations in the auciton.
\end{theorem}

\begin{proof}
    We will assume that $n \ge 2$;
    otherwise the lower bound degenerates to $\Omega(\epsilon^2)$ which follows from the lower bound of the single-bidder case (c.f., \citet*{HuangMR:SICOMP:2018}).
    This sample complexity lower bound is based on the same set of feasible allocations as in Theorem~\ref{thm:revenue-lipschitz-lower-bound}, and slightly modified value distributions.
    We start by restating the set of feasible allocations:
    \[
        \allocset = \Big\{ \big(0, 0, \dots, 0\big), \big(\,\frac{k}{n}, \frac{k}{n}, \dots, \frac{k}{n}\,\big) \Big\}
        ~.
    \]
    In other words, we allocate either to none of the bidders, or to all of them each by a $\frac{k}{n}$ amount.

    We next define the value distributions.
    The bidders' values are binary, either $0$ or $1$.
    There are two possible value distributions $D^+$ and $D^-$.
    In the first value distribution $D^+$ (of a bidder), the value of each bidder is $0$ with probability $\frac{1}{n}-\delta$ and is $1$ with probability $1-\frac{1}{n}+\delta$, where $\delta \le \frac{1}{2n}$ is a parameter to be determined later in the analysis.
    In second value distribution $D^-$, the value of each bidder is $0$ with probability $\frac{1}{n}+\delta$, and is $1$ with probability $1-\frac{1}{n}-\delta$.

    It is straightforward to upper bound the \emph{Hellinger distance} $H(D^+, D^-)$ of the distributions.
    We present both the definition (of the squared Hellinger distance) and the calculation below to be self-contained:
    \begin{align*}
        H^2(D^+,D^-) 
        &
        \defeq \frac{1}{2} \sum_{v \in \{0, 1\}} \big(\sqrt{\Pr_{D^+}[v]} - \sqrt{\Pr_{D^-}[v]} \big)^2 \\
        &
        = \frac{1}{2} \Big( \sqrt{\frac{1}{n}-\delta}-\sqrt{\frac{1}{n}+\delta} \Big)^2 + \frac{1}{2} \Big( \sqrt{1-\frac{1}{n}+\delta}-\sqrt{1-\frac{1}{n}-\delta} \Big)^2 \\[1.5ex]
        &
        = \frac{2\delta^2}{\big( \sqrt{\frac{1}{n}-\delta}+\sqrt{\frac{1}{n}+\delta} \, \big)^2} + \frac{2\delta^2}{\big( \sqrt{1-\frac{1}{n}+\delta} + \sqrt{1-\frac{1}{n}-\delta} \, \big)^2}\\
        &
        \le \frac{4 \delta^2}{\big( \sqrt{\frac{1}{n}-\delta}+\sqrt{\frac{1}{n}+\delta} \, \big)^2}
        && \mbox{($n \ge 2$)}
        \\[1ex]
        &
        \le 2 \delta^2 n
        ~.
        && \mbox{($\sqrt{x} + \sqrt{y} \ge \sqrt{x+y}$)}
    \end{align*}



    We shall use two properties of the Hellinger distance, which we state below without proofs.
    We refer interested readers to \citet{GuoHTZ:COLT:2021} for properties of the Hellinger distance and their applications in proving sample complexity upper and lower bounds of auctions and other optimization problems.

    \begin{lemma}[c.f., Lemma 3 of \citet{GuoHTZ:COLT:2021}]
        \label{lem:hellinger-subadditive}
        For any product distributions $\bD = D_1 \times \dots D_m$ and $\btildeD = \tilde{D}_1 \times \dots \times \tilde{D}_m$ we have:
        \[
            H^2 \big( \bD, \btildeD \big) ~\le~ \sum_{i=1}^m H^2 \big( D_i, \tilde{D}_i \big)
            ~.
        \]
    \end{lemma}

    \begin{lemma}[c.f., Lemmas 1 and 2 of \citet{GuoHTZ:COLT:2021}]
        \label{lem:hellinger-expectation}
        For any distributions $D$ and $\tilde{D}$ over a common domain $\Omega$, and any function $f : \Omega \to [0, 1]$, the expectations $f(D)$ and $f(\tilde{D})$ differ by at most:
        \[
            \sqrt{2} \cdot H \big( D, \tilde{D} \big)
            ~.
        \]
    \end{lemma}

    We further examine the bidders' virtual values.
    A bidder's virtual value would be $1$ w.r.t.\ both distributions when its value is $1$.
    If a bidder's value is $0$, on the other hand, the value values are:
    \[
        \virtual^+(0) = - \frac{1-\frac{1}{n} + \delta}{\frac{1}{n} - \delta} = - (n-1) - \frac{n \delta}{\frac{1}{n} - \delta} < - (n-1) - \frac{n^2 \delta}{2} ~,
    \]
    w.r.t.\ distribution $D^+$, and:
    \[
        \virtual^-(0) = - \frac{1-\frac{1}{n} - \delta}{\frac{1}{n} + \delta} = - (n-1) + \frac{n \delta}{\frac{1}{n} + \delta} > - (n-1) + \frac{n^2 \delta}{2}
    \]
    w.r.t.\ distribution $D^-$.

    We will choose $\delta = \frac{48 \epsilon}{nk}$ which is indeed at most $\frac{1}{2n}$ since $\epsilon \le \frac{1}{100}$.
    By the above bounds:
    \begin{enumerate}
        \item $H^2(D^+, D^-) \le \frac{4608\epsilon^2}{nk^2}$;
        \item $\virtual^+(0) < -(n-1)-\frac{24\epsilon n}{k}$;
        \item $\virtual^-(0) > -(n-1)+\frac{24\epsilon n}{k}$.
    \end{enumerate}

    %
    If all $n$ bidders have value $1$, it is clear that we should allocate a $\frac{k}{n}$ to all $n$ bidders and collect $k$ in virtual welfare, regardless of whether each bidder's value distribution is $D^+$ or $D^-$.
    If only $n-1$ bidders have value $1$ and a bidder $i$ has value $0$, however, the virtual welfare maximizing allocation depends on bidder $i$'s value distribution.
    If bidder $i$'s distribution is $D^+$, the virtual welfare of allocating to all bidders is at most:
    \[
        \frac{k}{n} \big( (n-1) \cdot 1 + 1 \cdot \virtual^+(0) \big) < - 24 \epsilon
        ~.
    \]
    Hence, it is better to allocate to none of the bidders and get $0$ virtual welfare.
    Similarly, if bidder $i$'s distribution is $D^-$, the virtual welfare of allocating to all bidders is more than $24 \epsilon$, which is better than allocating to none of the bidders.
    In either case, making the wrong choice would lose more than $24 \epsilon$ in the virtual welfare compared to the optimal auction.
    
    \paragraph{Intuition.}
    With only $N \le \frac{c n k^2}{\epsilon^2}$ samples for a sufficiently small constant $c$, we have:
    \[
        H^2\big( (D^+)^N, (D^-)^N \big) \le N \cdot H^2\big( D^+, D^- \big) \ll 1
    \]
    by Lemma~\ref{lem:hellinger-subadditive} and by the first property above. 
    Then as a corollary of Lemma~\ref{lem:hellinger-expectation}, an algorithm with only $N$ samples must essentially choose the same allocation when $i$ is the only bidder with value $0$ regardless of whether $i$'s value distribution is $D^+$ or $D^-$. 
    This means that the algorithm would be wrong with at least a constant probability, say, $\frac{1}{3}$.

    Further, the probability of having such a vector profile is at least:
    \begin{equation}
        \label{eqn:sample-complexity-lower-bound-bad-prob}
        \Big( 1-\frac{1}{n}-\delta \Big)^{n-1} \Big( \frac{1}{n} - \delta \Big)
        \ge \Big( 1-\frac{3}{2n} \Big)^{n-1} \frac{1}{2n} \ge \frac{1}{8n}
        ~.
    \end{equation}
    
    Hence, the loss in expected revenue due to making the wrong decision when bidder $i$ is the only bidder with value $0$ is more than:
    \[
        \frac{1}{8n} \cdot \frac{1}{3} \cdot 24 \epsilon = \frac{\epsilon}{n}
        ~.
    \]
    
    Since the argument applies to all $n$ bidders, the total revenue loss compared to the optimal revenue is more than $\epsilon$.
    
    \paragraph{Formal Proof.}
    Consider the following set of $2^n$ product distributions:
    \[
        \mathcal{D} = \Big\{ \, \bD : D_i = D^+ ~\text{or}~D^- \, \Big\}
        ~.
    \]

    We will prove that the average difference between the optimal revenue and the expectation of the algorithm's auction's revenue is more than $\epsilon$, when the underlying product value distribution is chosen uniformly from $\mathcal{D}$.

    We first define some notations.
    For any product value distribution $\bD \in \mathcal{D}$, let $\OPT_{\bD}$ denote Myerson's optimal auction w.r.t.\ $\bD$, and let $\ALG_{\bD}$ denote the auction selected by an algorithm with $N \le \frac{c nk}{\epsilon^2}$ i.i.d.\ samples from $\bD$.
    For any value profile $\bv$, we further let $\ALG_{\bD}(\bv)$ denote the allocation of $\ALG_{\bD}$ when the values are $\bv$.
    Let $\Dif_{\bD} (\bv)$ be the difference between the maximum virtual welfare given by $\OPT_{\bD}$ when the bidders' values are $\bv$, and the virtual welfare given by $\ALG_{\bD}$, taking expectation over the randomness of the $N$ i.i.d.\ samples and the intrinsic randomness of the algorithm.
    For any agent $i$, let $\bv^i$ denote the value profile in which agent $i$ is the only agent with value $0$, i.e., $v^i_i = 0$ and $v^i_j = 1$ for any $j \ne i$.

    For any product value distribution $\bD \in \mathcal{D}$ we have:
    \begin{align*}
        \E_{\bv \sim \bD} \, \Dif_{\bD}(\bv)
        &
        ~\ge~ \sum_{i = 1}^n \, \Pr_{\bv \sim D} \big[ \bv = \bv^i \big] \cdot \Dif_{\bD}(\bv^i)
        \\
        &
        ~\ge~ \frac{1}{8n} \sum_{i = 1}^n \, \Dif_{\bD}(\bv^i)
        ~.
        && \mbox{(Eqn.~\eqref{eqn:sample-complexity-lower-bound-bad-prob})}
    \end{align*}

    Averaging over all distributions $\bD \in \mathcal{D}$, we get that:
    \[
        \frac{1}{2^n} \sum_{\bD \in \mathcal{D}} \E_{\bv \sim \bD}  \, \Dif_{\bD}(\bv)
        ~\ge~
        \frac{1}{2^{n+3} n} \sum_{\bD \in \mathcal{D}} \sum_{i = 1}^n \Dif_{\bD}(\bv^i)
        ~.
    \]

    Next we change the order of summation, and for each $i$ we pair up product value distributions that only differ in agent $i$'s distribution.
    The right-hand-side is then rewritten as:
    \[
        \frac{1}{2^{n+3} n} \sum_{i=1}^n \sum_{\bD_{-i} \in \{D^+, D^-\}^{n-1}} \Big( \Dif_{(D_i = D^+, \bD_{-i})}(\bv^i)
        + \Dif_{(D_i = D^-, \bD_{-i})}(\bv^i) \Big)
        ~.
    \]

    To prove that the above equation is greater than $\epsilon$, it suffices to show that for any bidder $i$ and any $D_{-i} \in \{ D^+, D^- \}^{n-1}$, we have:
    \[
        \Dif_{(D_i = D^+, \bD_{-i})}(\bv^i)
        + \Dif_{(D_i = D^-, \bD_{-i})}(\bv^i) > 16 \epsilon
        ~.
    \]

    On the one hand:
    \[
        \Dif_{(D_i = D^+, \bD_{-i})}(\bv^i) ~>~ 24 \epsilon \cdot \Pr \Big[ \ALG_{(D_i = D^+, \bD_{-i})}(v^i) = \big(\,\frac{n}{k}, \cdots, \frac{n}{k}\,\big) \Big]
        ~.
    \]

    On the other hand:
    \[
        \Dif_{(D_i = D^-, \bD_{-i})}(\bv^i) ~>~ 24 \epsilon \cdot \Pr \Big[ \ALG_{(D_i = D^-, \bD_{-i})}(v^i) = (0, \cdots, 0) \Big]
        ~.
    \]

    Our task therefore becomes proving that:
    \[
        \Pr \Big[ \ALG_{(D_i = D^+, \bD_{-i})}(v^i) = \big(\,\frac{n}{k}, \cdots, \frac{n}{k}\,\big) \Big] + \Pr \Big[ \ALG_{(D_i = D^-, \bD_{-i})}(v^i) = (0, \cdots, 0) \Big] \ge \frac{2}{3}
        ~.
    \]

    Note that:
    \[
        \Pr \Big[ \ALG_{(D_i = D^-, \bD_{-i})}(v^i) = \big(\,\frac{n}{k}, \cdots, \frac{n}{k}\,\big) \Big] + \Pr \Big[ \ALG_{(D_i = D^-, \bD_{-i})}(v^i) = (0, \cdots, 0) \Big] = 1
        ~.
    \]
    
    We just need to show that:
    \[
        \Pr \Big[ \ALG_{(D_i = D^-, \bD_{-i})}(v^i) = \big(\,\frac{n}{k}, \cdots, \frac{n}{k}\,\big) \Big] - \Pr \Big[ \ALG_{(D_i = D^+, \bD_{-i})}(v^i) = \big(\,\frac{n}{k}, \cdots, \frac{n}{k}\,\big) \Big] \le \frac{1}{3}
        ~.
    \]

    It follows by Lemma~\ref{lem:hellinger-expectation}, with function $f$ taking the $N$ i.i.d.\ samples as input and returning the probability that the auction would choose $(\frac{n}{k}, \cdots, \frac{n}{k})$ when the bidders' values are $\bv^i$, and by the upper bound of the Hellinger distance:
    \[
        H^2 \big( (D_i = D^+, \bD_{-i})^N, (D_i = D^-, \bD_{-i})^N \big) \le N \cdot H^2(D^+, D^-) \ll 1
        ~.
    \]
\end{proof}

\paragraph{Beyond Unit Demand.}
Recall the reduction that we mentioned in Section~\ref{sec:prelim}:
an $\epsilon$-approximation in an auction with $n$ bidders, rank $k$, and maximum demand $d$ is equivalent to an $\frac{\epsilon}{d}$-approximation in another auction with $n$ bidders, rank $\frac{k}{d}$, and unit-demand, obtained by scaling all allocations by a factor $\frac{1}{d}$.
Using this reduction we get the following corollaries when bidders' maximum demand $d$ can be larger than $1$.
Note that $d \le k$ and thus $d$ is subsumed by $k$ inside the logarithmic factors.

\begin{corollary}[Downward-closed Feasibility, General Demand]
    \label{cor:sample-complexity}
    Consider an arbitrary downward-closed single-parameter auction with $n$ bidders, rank $k$, and maximum demand $d$ of any single bidder.
    Suppose that we have $N$ i.i.d.\ samples where $N$ is at least:
    \[
        C \cdot \frac{nkd}{\epsilon^2} \log \frac{nk}{\epsilon\delta}
    \]
    for a sufficiently large constant $C$.
    Then Myerson's optimal auction w.r.t.\ the dominated product empirical distribution $\btildeE$ is an $\epsilon$-approximation with probability at least $1-\delta$.
\end{corollary}

\begin{corollary}[General Feasibility, General Demand]
    \label{cor:sample-complexity-general}
    Consider an arbitrary single-parameter auction with $n$ bidders, rank $k$, and maximum demand $d$ of any single bidder.
    Suppose that we have $N$ i.i.d.\ samples where $N$ is at least:
    \[
        C \cdot \frac{nk^2}{\epsilon^2} \log \frac{nk}{\epsilon} \log \frac{nk}{\epsilon\delta}
    \]
    for a sufficiently large constant $C$.
    Then Myerson's optimal auction w.r.t.\ the dominated product empirical distribution $\btildeE$ is an $\epsilon$-approximation with probability at least $1-\delta$.
\end{corollary}

Readers may notice that Corollary~\ref{cor:sample-complexity-general} gives the same sample complexity bound as Theorem~\ref{thm:sample-complexity-general}, independent of the maximum demand $d$.
This is not a typo but instead follows from the fact that $k$ and $\epsilon$ are homogeneous in the bound of Theorem~\ref{thm:sample-complexity-general}.
The reduction in Section~\ref{sec:prelim} scales both $k$ and $\epsilon$ by a factor $\frac{1}{d}$ and therefore the two effects cancel out.



\bibliographystyle{plainnat}
\bibliography{ref}

\end{document}